\documentclass[journal,twocolumn]{IEEEtran}
\usepackage{amsmath,epsfig}
\usepackage{graphics,dblfloatfix}
\usepackage{supertabular}
\usepackage{amsxtra}
\usepackage{graphicx}
\usepackage{amssymb}
\usepackage{algorithm}
\usepackage{algorithmic}
\usepackage{graphicx}
\usepackage{fixltx2e}
 \usepackage{bbm}
\providecommand{\norm}[1]{{\lVert#1\rVert}}
\providecommand{\sqnorm}[1]{{\lVert#1\rVert^{2}_{2}}}

\begin{document}

\title{Joint Reconstruction of Multi-view Compressed Images 
\thanks{Part of this work has been accepted to the European Signal Processing Conference (EUSIPCO), Bucharest, Romania, Aug. 2012 \cite{Vijay_EUSIPCO2012}. 
 }}
 
\author{Vijayaraghavan~Thirumalai, ~\IEEEmembership{Student Member,~IEEE,} and 
            Pascal~Frossard, ~\IEEEmembership{Senior Member,~IEEE} 
            \thanks{The authors are with Signal Processing Laboratory - LTS4, Institute
of Electrical Engineering, Ecole Polytechnique F\'ed\'erale de Lausanne (EPFL), Lausanne 1015, Switzerland. e-mail: (vijayaraghavan.thirumalai@epfl.ch; pascal.frossard@epfl.ch).}
}

\maketitle

\begin{abstract}

The distributed representation of correlated multi-view images is an important problem that arise in vision sensor networks. This paper concentrates on the joint reconstruction problem where the distributively compressed correlated images are jointly decoded in order to improve the reconstruction quality of all the compressed images. We consider a scenario where the images captured at different viewpoints are encoded independently using common coding solutions (e.g., JPEG, H.264 intra) with a balanced rate distribution among different cameras. A central decoder first estimates the underlying correlation model from the independently compressed images which will be used for the joint signal recovery. The joint reconstruction is then cast as a constrained convex optimization problem that reconstructs total-variation (TV) smooth images that comply with the estimated correlation model. At the same time, we add constraints that force the reconstructed images to be consistent with their compressed versions. We show by experiments that the proposed joint reconstruction scheme outperforms independent reconstruction in terms of image quality, for a given target bit rate. In addition, the decoding performance of our proposed algorithm compares advantageously to state-of-the-art distributed coding schemes based on disparity learning and on the DISCOVER.

\end{abstract}


\section{Introduction}
In recent years, vision sensor networks have been gaining an ever increasing popularity enforced by the availability of cheap semiconductor components. These systems usually acquire multiple correlated images of the same 3D scene from different viewpoints. Compression techniques shall exploit this correlation in order to efficiently represent the 3D scene information. The distributed coding paradigm becomes particularly attractive in such settings; it permits to efficiently exploit the correlation between images with low encoding complexity and minimal inter-sensor communication, which directly translate into
power savings in sensor networks. In the distributed compression framework, a central decoder jointly reconstructs the visual information from the compressed images by exploiting the correlation between the samples. This permits to achieve a good rate-distortion tradeoff in the representation of correlated multi-view images, even if the encoding is performed independently.

The first information-theoretic results on distributed source coding appeared in the late seventies for the
noiseless \cite{Slepian} and noisy cases \cite{WynerZiv}.  
 However, most results in distributed coding have remained non-constructive for about
three decades. Practical DSC schemes have then been designed by establishing a relation between the
Slepian-Wolf theorem and channel coding \cite{Pradhan}. Based on the results in \cite{Pradhan}, several distributed coding schemes for video and multi-view images have been proposed in the literature \cite{DVC_overview, DVC_overview_sp}. In such schemes, a feedback channel is generally used for accurately controlling the Slepian-Wolf coding rate. Unfortunately, this results in increased latency and bandwidth usage due to the multiple requests from the decoder. These schemes can thus hardly be used in real time applications.  One solution to avoid the feedback channel is to use a separate encoding rate control module to precisely control the Slepian-Wolf coding rate \cite{prism}.  The overall computational complexity at the encoder becomes non-negligible due to this rate control module. In this paper, we build a distributed coding scheme, where the correlated compressed images are directly transmitted to the joint decoder without implementing any Slepian-Wolf coding; this avoids the necessity for complex estimation of the statistical correlation estimation and of the coding rate at the encoder.

 \begin{figure*}[h!]
 \centerline{\epsfig{figure=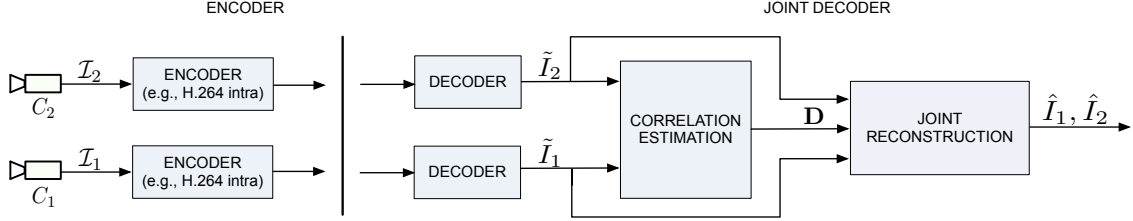,width=15.5cm}}
\caption{Schematic representation of our proposed framework. The images
$\mathcal{I}_1$ and $\mathcal{I}_2$ are correlated through displacement of scene
objects due to positioning of the cameras $C_1$ and $C_2$.}

\label{Fig:system}
\end{figure*}

We consider a scenario where a set of cameras are distributed in a 3D scene. In most practical deployments of such
systems, the images captured by the different cameras are likely to be correlated.
The captured images are encoded independently using standard encoding solutions and are transmitted to the central decoder. Here, we assume that the images are compressed using balanced rate allocation, which permits to share the transmission and computational costs equally among the sensors. It thus prevents the necessity for hierarchical relationship among the sensors. The central decoder builds a correlation model from the compressed images which is used to jointly decode the multi-view images. The joint reconstruction is formulated as a convex optimization problem. It reconstructs the multi-view images that are consistent with the underlying correlation information and with the compressed images information.  While reconstructing the images, we also effectively handle the occlusions that commonly arise in multi-view imaging. We solve the joint reconstruction problem using effective parallel proximal algorithms \cite{Combettes_prox}. 

We evaluate the performance of our novel joint decoding scheme in several multi-view datasets.  Experimental results demonstrate that the proposed distributed coding solution improves the rate-distortion performance of the separate coding results by taking advantage of the inter-view correlation. We show that the quality of the decoded images is quite balanced for a given bit rate, as expected from a symmetric coding solution. We observe that our scheme, at low bit rate, performs close to the joint encoding solutions based on H.264, when the block size used for motion compensation is set to $4\times4$.  Finally, we show that our framework outperforms state-of-the-art distributed coding solutions based on disparity learning \cite{David} and on the DISCOVER codec \cite{discover}, in terms  of rate-distortion performance. It certainly provides an interesting alternative to most classical DSC solutions \cite{DVC_overview, DVC_overview_sp,prism}, since it does not require any statistical correlation information at the encoder. 

Only very few works in the literature address the distributed compression problem without using a channel encoder or a feedback channel. In \cite{Wag}, a distributed coding technique for compressing the multi-view images has been proposed, where a joint decoder reconstructs the views from low resolution images using super-resolution techniques. In more details, each sensor transmits a low resolution compressed version of the original image to the decoder. At the decoder, these low resolution images are registered with respect to a reference image, where the image registration is performed by shape analysis and image warping. The registered low resolution images are then jointly processed to decode a high resolution image using image super-resolution techniques. However, this framework requires communication between the encoders in order to facilitate the registration, e.g., the transmission of feature points.  Other works in super-resolution use multiple compressed images that are fused for improved resolution \cite{sr_overview}. Such techniques usually target reconstruction of a \emph{single} high resolution image from multiple compressed images. Alternatively,  techniques have been developed in \cite{mdc, mdc1} to decode a \emph{single} high quality image from several encoded versions of the same source image or videos. This is achieved by solving an optimization problem that enforces the final reconstructed image to be consistent with all the compressed copies.  Our main target in this paper is to jointly improve the quality of \emph{multiple} compressed correlated (multi-view) images and not to increase the spatial resolution of the compressed images or to extract a single high quality image. 
More recently, Schenkel {\emph{et al.}} \cite{Schenkel} have considered a distributed representation of image pairs. In particular, they have proposed an optimization framework to enhance the quality of the JPEG compressed images. This work,  however, considered an asymmetric scenario that requires a reference image for joint decoding. 

The rest of the paper is organized as follows. The joint decoding algorithm along with the optimization framework for joint reconstruction is described in Section~\ref{sec:joint_decoder}. In Section~\ref{Sec:optmeth}, we present the optimization algorithm based on proximal splitting methods. In Section~\ref{sec:results},  we present the experimental results for the joint reconstruction of pairs of images. Section~\ref{sec:multiview} describes the extension of our proposed framework to decode multiple images along with the simulation results. Finally, in Section~\ref{sec:conc} we draw some concluding remarks.

\section{Joint Decoding of Image Pairs} \label{sec:joint_decoder}

 We consider the scenario illustrated in Fig.~\ref{Fig:system}, where a pair of cameras $C_1$ and $C_2$ project the 3D visual information on the 2D plane $\mathcal{I}_1$ and $\mathcal{I}_2$ (with resolution $N = N_1\times N_2$), respectively. The images $\mathcal{I}_1$ and $\mathcal{I}_2$ are compressed independently using standard encoding solutions (e.g., JPEG, H.264 intra) and are transmitted to a central decoder. The joint decoder has the access to the compressed version of the correlated images and its main objective is to improve the quality of all the compressed views by exploiting the underlying inter-view correlation.  We first propose to estimate the correlation between images from the decoded images $\tilde{I}_1$ and $\tilde{I}_2$, which is effectively modeled by a dense depth image $D$. 
The joint reconstruction stage then uses the depth information $D$ and enhances the quality of the decoded images $\tilde{I}_1$ and $\tilde{I}_2$. 
Note that one could solve a joint problem to estimate simultaneously the correlation information $D$ and the improved images. However, such a joint optimization problem would be hard to solve with a complex objective function. Therefore, we propose to split the problem in two steps: (i) we estimate a correlation information from the decoded images; and (ii) we carry out joint reconstruction using the estimated correlation information. These two steps are detailed in the rest of this section.

\subsection{Depth Estimation} \label{sec:depth_est}
The first task is to estimate the correlation between images, which typically consists in a depth image. In general, the dense depth information is estimated by matching the corresponding pixels between images. Several algorithms have been proposed in the literature to compute dense depth images. For more details, we refer the reader to \cite{Scharstein}. In this work, we estimate a dense depth image from the compressed images in a regularized energy minimization framework, where
the energy $E$ is composed of a data term $E_d$ and a smoothness term $E_s$. A dense depth image $D$ is
obtained by minimizing the energy function $E$ as
\begin{equation} \label{eqn:energy_chap4}
D= \underset{D_c}{\operatorname{argmin}} \; E(D_c) =  \underset{D_c}{\operatorname{argmin}} \;  \{E_d(D_c) + \lambda \; E_s(D_c)\},
\end{equation}
where $\lambda$ balances the importance of the data and smoothness terms, and  $D_c$ represents the candidate depth images. The candidate depth values $D_c(m,n)$ for every pixel position $(m,n)$ are discrete; this is constructed by uniformly sampling the inverse depth in the range $[1/D_{max}, 1/D_{min}]$, where $D_{min}$ and $D_{max}$ are the minimal and maximal depth values in the scene, respectively \cite{multiview_gc}.

We now discuss in more details the components of the energy function of Eq.~(\ref{eqn:energy_chap4}). The data term, $E_d$ is used to match the pixels across views by assuming that the 3D scene surfaces are Lambertian, i.e., the intensity is consistent irrespective of the viewpoints. It is computed as 
 \begin{equation} \label{eqn:datacost}
E_d(D_c) =  \sum_{m=1}^{N_1}  \sum_{n=1}^{N_2} \mathcal{C}((m,n),D_c(m,n)),
\end{equation}
where $N_1$ and $N_2$ represent the image dimensions and $(m,n)$ represent a pixel position. The most commonly used pixel-based cost function $\mathcal{C}$ includes squared intensity differences and absolute intensity differences. In this work, we use square intensity difference to measure the disagreement of assigning a depth value $D_c(m,n)$ to the pixel location $(m,n)$. Mathematically, it is computed as 
\begin{equation} \label{eqn:pixelcost}
\small
 \mathcal{C}((m,n),D_c(m,n)) =  \sqnorm{\tilde{I}_2(m,n)- \mathcal{W}(\tilde{I}_1(m,n),D_c(m,n))}, 
\end{equation}
\normalsize
where $\mathcal{W}$ is a warping function that warps the image $\tilde{I}_1$ using the depth value $D_c(m,n)$. This warping, in general, is a two step process~\cite{mv_geo}. First the pixel position $(m,n)$ in the image $\tilde{I}_1$ is projected to the world coordinate system. This projection step is represented as 
\begin{equation} \label{eqn:proj_step1}
[u, v, w]^T = R_1 P_1^{-1} [m,n,1]^T D_c(m,n) + T_1,
\end{equation} 
where $P_1$ is the intrinsic camera matrix of the camera $C_1$ and $(R_1, T_1)$ represent the extrinsic camera parameters with respect to the global coordinate system. Then, the 3D point $[u, v, w]^T$ is projected on the coordinates of the camera $C_2$ with the internal and external camera parameters, respectively as $P_2$ and $(R_2,T_2)$. This projection step can be described as 
\begin{equation}\label{eqn:proj_step2} 
[x^{\prime}, y^{\prime}, z^{\prime}]^T = P_2 R_2^{-1} \{[u,v,w]^T - T_2\}.
\end{equation}
Finally, the pixel location of the warped image is taken as $(m^\prime, n^\prime)=(round(x^{\prime}/z^{\prime}), (round(y^{\prime}/z^{\prime}))$, where $round(x)$ rounds $x$ to the 
nearest integer. 

The smoothness term, $E_s$ is used to enforce consistent depth values at neighboring pixel locations $(m,n)$ and $(\tilde{m}, \tilde{n})$. It is measured as 
\begin{equation}\label{eqn:chap4_energy}
  E_s(D_c) = \sum_{(m,n),(\tilde{m},\tilde{n})\in \mathcal{N}} min( |D_c(m,n)-D_c(\tilde{m}, \tilde{n})|, \tau),
\end{equation}
where $\mathcal{N}$ represents the usual four-pixel neighborhood and $\tau$ sets an upper level on the smoothness penalty such that discontinuities can be preserved \cite{Veksler}. 

We can finally rewrite the regularized energy objective function for the depth estimation problem as
\begin{equation} \label{eqn:depth_est}
\begin{aligned}
E(D_c) = & \sum_{m=1}^{N_1}  \sum_{n=1}^{N_2} \mathcal{C}((m,n),D_c(m,n)) +  \\
&  \lambda \sum_{(m,n),(\tilde{m}, \tilde{n}) \in \mathcal{N}} min( |D_c(m,n)-D_c(\tilde{m}, \tilde{n})|, \tau).
  \end{aligned}
\end{equation}
This cost function is used in the optimization problem of Eq.~(\ref{eqn:energy_chap4}), which is usually a non-convex problem. Several minimization algorithms exist in the literature to solve Eq.~(\ref{eqn:energy_chap4}), e.g., Simulated annealing \cite{Barnard}, Belief Propagation \cite{Belief_prop}, Graph Cuts \cite{Graph_cuts, Kolomogorov}. Among these solutions, the optimization techniques based on Graph Cuts compute the minimum energy in polynomial time and they generally give better results than the other techniques \cite{Scharstein}. Motivated by this, in our work, we solve the minimization problem of Eq.~(\ref{eqn:energy_chap4}) using Graph Cut techniques.

\subsection{Image Warping as Linear Transformation }\label{sec:warp_as_linear}
Before describing our joint reconstruction problem, we show how the image warping operation ${\mathcal{W}(\tilde{I}_1,D)}$ in Eq.~(\ref{eqn:pixelcost}) can be written as matrix multiplication of the form $A\cdot \mathcal{R}(\tilde{I}_1)$\footnote{For consistency, we use the compressed image $\tilde{I}_1$; however, this matrix multiplication holds even if one uses the original image $\mathcal{I}_1$ for warping.}; this linear representation offers a more flexible formulation of our joint reconstruction problem. The reshaping operator ${\mathcal{R}}: I_{N_1\times N_2}\rightarrow X_{N_1N_2 \times 1}$ produces a vector $X = \mathcal{R}(I) = [ I_{.,1} ^T \; I_{.,2} ^T \ldots  I_{.,{N_1}}^T] ^T$ from the matrix $I$, where $I_{.,m}$ represents the $m^{th}$ row of the matrix $I$ and $(.)^T$ denotes the usual transpose operator. For our convenience, we also define another operator $\mathcal{R}^{-1}_{N_1\times N_2}: X_{N_1N_2 \times 1} \rightarrow I_{N_1\times N_2}$ that takes the vector $X = [\mathcal{R}(I)]_{N_1N_2 \times 1}$ and gives back the matrix $I_{N_1\times N_2}$, i.e., this operator $\mathcal{R}^{-1}$ performs the inverse operations corresponding to $\mathcal{R}$. 
The matrix $A$ describes the warping by re-arranging the elements of $\mathcal{R}(\tilde{I}_1)$. Its construction is described in this section.

 We have shown earlier that the warping function $\mathcal{W}$ shifts the pixel position $(m,n)$ in the reference image to the position $(m^\prime, n^\prime)$ in the target image. Alternatively, this pixel shift between images can be represented using a horizontal component $\bold{m}^h$ and a vertical component $\bold{m}^v$ of the motion field as $(m^\prime, n^\prime) = (m+ \bold{m}^h(m,n), n+\bold{m}^v(m,n))$. Note that this motion field $(\bold{m}^h,\bold{m}^v)$ can be easily computed from Eqs.~(\ref{eqn:proj_step1}) and (\ref{eqn:proj_step2}), once the depth information $D$ and the camera parameters are known. Now, our goal is to represent the motion compensation operation ${\tilde{I}_1(m+\bold{m}^h(m,n),n+\bold{m}^v(m,n))}$ as a linear transformation $A\cdot \mathcal{R}(\tilde{I}_1)$ given as 
\begin{eqnarray}\label{eqn:I2AI1}
\underbrace{\left[ \begin{array}{c}
                 \bar{I}_{2,1}^T  \\
                 \bar{I}_{2,2}^T  \\
                \vdots \\
                 \bar{I}_{2,{N_1}}^T  \end{array} \right] }_{\mathcal{R}(\bar{I}_2)}
                  =  \underbrace{\left[ \begin{array}{c}
                 A^1 \\
                 A^2  \\
                \vdots \\
                 A^{N_1} \end{array} \right] }_{A}
                 \underbrace{ \left[ \begin{array}{c}
                 \tilde{I}_{1,1} ^T \\
                 \tilde{I}_{1,2} ^T  \\
                \vdots \\
                 \tilde{I}_{1,{N_1}}^T  \end{array} \right]}_{\mathcal{R}(\tilde{I}_1)}.
                 \end{eqnarray}
Here, ${\bar{I}_2 = \mathcal{W}(\tilde{I}_1(m,n),D)}$ represents the warped image and $A^m$ is a matrix of dimensions $N_2 \times N_1N_2$ whose entries  are determined by the horizontal and vertical components of the motion field in the $m^{th}$ row, i.e., $\bold{m}^h(m,.)$ and $\bold{m}^v(m,.)$. 

In general, the elements of the matrix $A^m$ can be found in two ways: (i)  forward warping; and (ii) inverse warping. In this work, we propose to construct the matrix $A^m$ based on forward warping; this permits easier handling of the occluded pixels as shown later. Given a motion vector, the elements of the matrix $A^m$ are given by
\begin{equation}\label{eqn:matA}
A^m(n-\beta_1-\beta_2N_2, n) =    \left\{
                              \begin{array}{ll}
                             1 & \mbox{if \ } \bold{m}^h(m,n) = \beta_1,\\
                                  & \mbox{and } \bold{m}^v(m,n) = \beta_2, \\
                             0 &  \mbox{otherwise.} \end{array} \right.
\end{equation}
If $n-\beta_1-\beta_2N_2 < 0$ (e.g., at image boundaries), we set ${n-\beta_1-\beta_2N_2 = 1}$ so that the dimensions of the matrix $A^m$ stays $N_2 \times N_1N_2$. It should be noted that the matrix $A^m$ formed using Eq.~(\ref{eqn:matA}) contains multiple entries with values of `$1$' in each row. This is because several pixels in the source image can be mapped to the same location in the destination image during forward warping. In such cases, for a given row index $m$ we keep only the last `$1$' entry in the matrix $A^m$ while the remaining ones in the row are set to zero. This is motivated by the fact that, during forward warping when multiple source pixels are mapped to the same destination point $(m^\prime, n^\prime)$, the intensity value of the last source pixel is assigned to the destination pixel $(m^\prime, n^\prime)$ \footnote{We assume that the pixels are scanned from left to right and then top to bottom.}. 
Furthermore, it is interesting to note that some of the rows in the matrix $A^m$ do not contain any entry with value of `$1$', i.e., all entries in $m^{th}$ row of $A^m$ are zeros. This means that the set of pixel locations $\{j: j \in \mathcal{J}^m \}$ in the warped image $\bar{I}_{2,m}(j)$ has zero value, where $\mathcal{J}^m$ is the set of row indexes in the matrix $A^m$ that do not contain any entry with value of  `$1$'. These pixel positions represent holes in the warped image that define the occluded regions.  Finally, the $m^{th}$ row in the warped image is represented as 
\begin{equation}\label{}
\bar{I}_{2,m}(j) =    \left\{
                              \begin{array}{ll}
                             0 & \mbox{if \ } j \in \mathcal{J}^m \\
                            \tilde{I}_{1}(k, n) &  \mbox{if \ } A^m(j,(k-m)N_2+n) =1. \end{array} \right.
\end{equation}
Thus, it is clear that the matrix $A^m$ shifts the pixels in $\tilde{I}_{1}$ by the corresponding motion vector $(\bold{m}^h(m,.), \bold{m}^v(m,.))$ in order to form $\bar{I}_{2,m}$. In a similar way, we can construct the matrix $A^m,  \forall m \in \{1,2,\ldots,N_1 \}$, and thus we can represent the image warping ${\mathcal{W}(\tilde{I}_1(m,n),D)}$ as $A\cdot \mathcal{R}(\tilde{I}_1)$. Finally, note that similar operations can also be performed with an inverse mapping. For details related to the construction of the matrix $A^m$ based on inverse warping, we refer the reader to \cite[Ch.\ 6, p.\ 95]{Vijay_thesis}.

\subsection{Joint reconstruction} \label{sec:jr}
We now discuss our novel joint reconstruction algorithm that takes benefit of the estimated correlation information given by the matrix $A$ (or $D$) in order to reconstruct the images. We propose to reconstruct an image pair $(\hat{I}_1, \hat{I}_2)$ as a solution to the following optimization problem:
\begin{align} \label{eqn:jr}
(\hat{I}_1, \hat{I}_2) = \;  \underset{I_1, I_2 \in \mathbb{R}^{N_1\times N_2} }
{\operatorname{argmin}}  \;
& (\norm{I_1}_{TV} +  \norm{I_2}_{TV})  \\ \nonumber
  \mbox{s.t.} \;
 &  \norm{\mathcal{R}(I_1) -  \mathcal{R}(\tilde{I}_1)}_2 \leq  \epsilon_1,\\ \nonumber
 &  \norm{\mathcal{R}(I_2) - \mathcal{R}(\tilde{I}_2)}_2 \leq  \epsilon_1, \\ \nonumber
 &  \norm{\mathcal{R}(I_2) -  A\cdot \mathcal{R}(I_1)}_2^2 \leq \epsilon_2.
\end{align}
Here,  $\tilde{I}_1$ and $\tilde{I}_2$ represent the decoded views (see Fig.~\ref{Fig:system}) and $\norm{.}_{TV}$ represents the total-variation (TV) norm. The first two constraints of Eq.~(\ref{eqn:jr}) forces the reconstructed images $\hat{I}_1$ and $\hat{I}_2$ to be close to the respective decoded images $\tilde{I}_1$ and $\tilde{I}_2$. The last constraint encourages the reconstructed images to be consistent with the correlation information represented by $A$, i.e., the warped image $A\cdot \mathcal{R}(I_1)$ should be consistent with the image $\mathcal{R}(I_2)$. 
Finally, the TV prior term ensures that the reconstructed images $\hat{I}_1$ and $\hat{I}_2$ are smooth. In general, inclusion of the prior knowledge brings effective reduction in the search space, which leads to efficient optimization solutions. 
The optimization problem of Eq.~(\ref{eqn:jr}), therefore reconstructs a pair of TV smooth images that is consistent with both the compressed images and the correlation information. In our framework, we use the TV prior on the reconstructed images, however one could also use a sparsity prior that minimizes the $l_1$ norm of the coefficients in a sparse representation of the images \cite{Donoho, Candes}

In the above formulation, it is clear that we measure the correlation consistency of all the pixels in the image $\mathcal{R}(I_2)$ and the warped image $A\cdot\mathcal{R}(I_1)$. However, this assumption is not true in multi-view imaging scenarios, as there are often problems due to occlusions. This indicates that we need to consider only the pixels that appear in both the views and we need to ignore the holes in the warped image $A \cdot \mathcal{R}(I_1)$ while enforcing consistency between $\mathcal{R}(I_2)$ and $A\cdot \mathcal{R}(I_1)$. 
The positions of holes in the warped image $A\cdot\mathcal{R}(I_1)$ correspond to the row indexes in the matrix $A$ that do not contain any value of  `$1$', i.e., all entries in a given row are zero. Once these rows are identified, we simply ignore that contribution while we measure the correlation consistency between the images $\mathcal{R}(I_2)$ and $A\cdot \mathcal{R}(I_1)$. More formally, let $\mathcal{J} = \bigcup_{m=1}^{N_1} \mathcal{J}^m$ be the set of indexes of these rows. Let us denote a diagonal matrix $M$ that is formed as 
\begin{equation} \label{eqn:matrix_M}
M(j,j) = \left\{
               \begin{array}{ll}
                 0 & \mbox{if \ } j \in \mathcal{J} \\
                 1  &  \mbox{otherwise,} \end{array} \right.
\end{equation}
where $j = \{1,2,\ldots,N_1N_2\}$. For effective occlusion handling, the joint reconstruction problem of Eq.~(\ref{eqn:jr}) can be modified as 
\begin{align}  \tag{OPT-1}\label{eqn:jr_f}
(\hat{I}_1, \hat{I}_2) = \;
 \underset{I_1, I_2}{\operatorname{argmin}}  \; &
(\norm{I_1}_{TV} +  \norm{I_2}_{TV})  \\ \nonumber
  \mbox{s.t.} \;  
 & \norm{\mathcal{R}(I_1) -  \mathcal{R}(\tilde{I}_1)}_2 \leq  \epsilon_1, \\ \nonumber
 &  \norm{\mathcal{R}(I_2) - \mathcal{R}(\tilde{I}_2)}_2 \leq  \epsilon_1, \\ \nonumber
 &  \norm{M(\mathcal{R}(I_2) -  A\cdot\mathcal{R}(I_1))}_2^2 \leq \epsilon_2.
\end{align}
Note that, by setting $M = \mathbbm{1}$, we get the optimization problem of Eq.~(\ref{eqn:jr}) that considers the consistency of all the pixels in $\mathcal{R}(I_2)$ and $A\cdot \mathcal{R}(I_1)$. We show later that the quality of the reconstructed images are improved,  when our joint decoding problem OPT-1 is solved with the matrix $M$ constructed using Eq.~(\ref{eqn:matrix_M}). Finally, the depth estimation and the joint reconstruction steps could be iterated several times. In our experiments, however, we have not observed any  significant improvement in the quality of the reconstructed images by repeating these two steps. 

\section{Optimization methodology} \label{Sec:optmeth}

We propose now a solution for the joint reconstruction problem \ref{eqn:jr_f}. We first show that the optimization problem is convex. Then, we propose an effective solution based on proximal methods. 

\newtheorem{prop}{Proposition} 
\begin{prop} \label{prop:jr_cvx}
The OPT-1 optimization problem is convex.
\end{prop}
\begin{proof}
Our objective is to show that all the functions in \ref{eqn:jr_f} problem are convex. However, it is quite easy to check that the functions $\norm{I_j}_{TV}$ and $\norm{\mathcal{R}(I_j) -  \mathcal{R}(\tilde{I}_j)}_2$, $\forall j \in \{1,2\}$ are convex \cite{Boyd_CVX}. So, we have to show that the last constraint $\norm{M(\mathcal{R}(I_2) - A\cdot\mathcal{R}(I_1))}_2^2$ is a convex function.  

Let $g(\grave{I}_1,\grave{I}_2) =  \sqnorm{\grave{I}_2 - \grave{A} \grave{I}_1}$, where $\grave{I}_2 = M\cdot \mathcal{R}(I_2)$, $\grave{A} = MA$ and $\grave{I}_1 = \mathcal{R}(I_1)$. 
The function $g$ can be represented as  
\begin{eqnarray} \nonumber 
g(\grave{I}_1,\grave{I}_2)   & = &(\grave{I}_2-\grave{A}\grave{I}_1)^T (\grave{I}_2-\grave{A}\grave{I}_1) \\  \nonumber %\label{eqn:f}
              & = & {\grave{I}_2^T\grave{I}_2- \grave{I}_2^T\grave{A}\grave{I}_1 - \grave{I}_1^T\grave{A}^T\grave{I}_2 + \grave{I}_1^T\grave{A}^T\grave{A}\grave{I}_1}.              
\end{eqnarray}
The second derivative $\nabla^2 g$ of the function $g$ is given as  
\begin{eqnarray} \nonumber 
  \nabla^2 g =  \left[   \begin{array}{cc}
                                     2\grave{A}\grave{A}^T & -2\grave{A} \\
                                     -2\grave{A}^T  & 2 \\
                                     \end{array}                                 
                                      \right]   
                 =  2C^TC   
                 \succeq 0.
\end{eqnarray}
Here, $C = [\grave{A}^{T} \; -\mathbbm{1}]$, where $\mathbbm{1}$ represents the identity matrix and $2C^{T}C \succeq 0$ follows from $2x^{T}C^{T}Cx$ = $2 \norm{Cx}_2^2 \geq 0$ for any $x$. This means that the Hessian function $\nabla^2 g$ is positive semi-definite and thus $g(\grave{I}_1,\grave{I}_2)$ is convex. 
\end{proof}

We now propose an optimization methodology to solve \ref{eqn:jr_f} convex problem with proximal splitting methods \cite{Combettes_prox}.  For mathematical convenience, we rewrite \ref{eqn:jr_f} as 
\begin{equation} \label{eqn:jr_mod_chap4}
\begin{aligned}
& \underset{X \in \mathbb{R}^{2N} }{\operatorname{argmin}}  
& &  \{ \norm{\mathcal{R}^{-1}(S_1X)}_{TV} +  \norm{\mathcal{R}^{-1}(S_2X)}_{TV}\} \\
& \mbox{s.t.} 
& & {\norm{S_1(Y - X)}_2 \leq  \epsilon_1},  {\norm{S_2(Y - X)}_2 \leq  \epsilon_1}, \\
&&&{\norm{BX}_2^2 \leq \epsilon_2},
\end{aligned}
\end{equation}
where $X = [  \mathcal{R}(I_1) \; ;   \mathcal{R}(I_2) ] $, $Y= [   \mathcal{R}(\tilde{I}_1) \; ;   \mathcal{R}(\tilde{I}_2) ] $, $S_1=[ \mathbbm{1} \; 0]$, $S_2=[ 0 \; \mathbbm{1}]$, 
$B=[  -MA \;  M ] $ and $\mathbbm{1}$ represents the identity matrix. Recall that $\mathcal{R}^{-1}_{N_1\times N_2}$ (for simplicity we omit the subscript in Eq.~(\ref{eqn:jr_mod_chap4})) is the operator that outputs a matrix of dimensions $N_1\times N_2$ from a column vector of dimensions $N = N_1N_2$. The optimization problem of Eq.~(\ref{eqn:jr_mod_chap4}) can be visualized as a special case of general convex problem as 
\begin{equation} \label{eqn:convex_opt_chap4}
\underset{X \in \mathbb{R}^{2N}}{\operatorname{argmin}} \{ f_1(X) + f_2(X) + f_3(X) + f_4(X) +  f_5(X)\},
\end{equation}
where the functions $f_1, f_2, f_3, f_4, f_5 \in \Gamma_0(\mathbb{R}^{2N})$ \cite{Combettes_prox}.  $\Gamma_0(\mathbb{R}^{2N})$ is the class of lower semicontinuous convex functions from $\mathbb{R}^{2N}$ to $(-\infty~+\infty]$ that are not infinity everywhere.  For the optimization problem given in Eq.~(\ref{eqn:jr_mod_chap4}), the functions in the representation of Eq.~(\ref{eqn:convex_opt_chap4}) are 
\begin{enumerate} 
\item{$f_1(X) =\norm{\mathcal{R}^{-1}(S_1X)}_{TV}$},
\item{ $f_2(X) =\norm{\mathcal{R}^{-1}(S_2 X)}_{TV}$},
\item{$ f_3(X) = i_{c_1}(X) = \left\{ \begin{array}{ll}  0 & X \in {c_1} \\ \infty & \mbox{otherwise,} \end{array} \right.
 $
\\ i.e., $f_3(X)$ is the indicator function of the closed convex set ${c_1} = \{ X  : \norm{S_1(Y-X)}_2 \leq  \epsilon_1 \},$ }
\item{$ f_4(X) = i_{c_2}(X) = \left\{ \begin{array}{ll}  0 & X \in {c_2} \\ \infty & \mbox{otherwise,} \end{array} \right.
 $
\\ where ${c_2} = \{ X  : \norm{S_2(Y-X)}_2 \leq  \epsilon_1 \},$ }
\item{ $f_5(X) = i_{c_3}(X) = \left\{ \begin{array}{ll} 0 & X \in {c_3} \\ \infty & \mbox{otherwise,} \end{array} \right. $
 \\ where $ {c_3} = \{ X  : \sqnorm{BX}  \leq \epsilon_2\}.$}     
\end{enumerate}

The solution to the problem of Eq.~(\ref{eqn:convex_opt_chap4}) can be estimated by generating the recursive sequence $X^{(t+1)} = prox_{f}(X^{(t)})$, where  the function $f$ is given as $f = \sum_{i=1}^5 f_i$. The proximity operator is defined as the $prox_f(X) = min_{X}~ \{ f (X)  + \frac{1}{2}\sqnorm{X-Z} \}$.  The main difficulty with these iterations is the computation of the $prox_{f}(X)$ operator. There is no closed form expression to compute the $prox_{f}(X)$, especially when the function $f$ is the cumulative sum of two or more functions.  In such cases, instead of computing the $prox_{f}(X)$ directly for the combined function $f$, one can perform a sequence of calculations involving separately the individual operators $prox_{f_i}(X), \forall i \in \{1,\ldots,5 \}$. The algorithms in this class are known as \emph{splitting methods} \cite{Combettes_prox}, which lead to an easily implementable algorithm. 

We describe in more details the methodology to compute the \emph{prox} for the functions $f_i, \forall i \in \{1,\ldots,5 \}$. 
For the function $f_1(X) = \norm{\mathcal{R}^{-1}(S_1X)}_{TV}$, the $prox_{f_1}(X)$ can be computed using Chambolle's algorithm \cite{Chambolle}. A similar approach can be used to compute the $prox_{f_2}(X)$.  The function $f_3$ can be represented  as $f_3 = F\circ G$, where $F = i_{d(\epsilon_1)}$ and $G = S_1X - S_1Y $. The set $d(\epsilon_1)$ represents the $l_2$-ball defined as $d(\epsilon_1) = \{ y  \in \mathbb{R}^{2N} : \norm{y}_2 \leq \epsilon_1 \} $. Then, the $prox_{f_3}$ can be computed using the following closed form expression:
\begin{equation} \label{eqn:chap4_proxf3}  
prox_{f_3}(X) = prox_{F\circ G}(X) = X + (S_1)^*(prox_{F}-\mathbbm{1})(G(X))
\end{equation}
\cite{Fadili_ICIP}, where $(S_1)^*$ represents the conjugate transpose of $S_1$. The $prox_{F}(y)$ with $F = i_{d(\epsilon_1)}$ can be computed using radial projection \cite{Combettes_prox} as  
\begin{equation} \label{eqn:chap4_l2ball}
prox_F(y) = \left\{ \begin{array}{ll}  y & \norm{y}_2 \leq \epsilon_1 \\ \frac{y}{\norm{y}_2} & \mbox{otherwise.} \end{array} \right.
\end{equation}
The $prox$ for the function $f_4$ can also be solved using Eq.~(\ref{eqn:chap4_proxf3}) by setting $F = i_{d(\epsilon_1)}$ and $G = S_2X- S_2Y $. Finally, the function $f_5$ can be represented with $F = i_{d(\sqrt{\epsilon_2})}$ and an affine operator $G_1 = BX $, i.e., $f_5 = F\circ G_1$. As the operator $B$ is not a tight frame, the $prox_{f_5}$ can be computed using an iterative scheme \cite{Fadili_ICIP}. Let $\mu_t \in ( 0, 2/\gamma_2 )$, and $\gamma_1$ and $\gamma_2$ be the frame constants with $\gamma_1 \mathbbm{1}  \leq B B^* \leq \gamma_2  \mathbbm{1}$. The $prox_{f_5}$ can be calculated iteratively \cite{Fadili_ICIP} as 
\begin{eqnarray}  \label{eqn:chap4_proxf4_step1}
u^{(t+1)} &=& \mu_t ( \mathbbm{1} - prox_{\mu_t^{-1}F})(\mu_t^{-1}u^{(t)} + G_1p^{(t)} ) \\ \label{eqn:chap4_proxf4_step2}
p^{(t+1)} &=& X -B^* u ^{(t+1)},
\end{eqnarray}
where $u^{(t)} \rightarrow u$ and $p^{(t)} \rightarrow prox_{F\circ G} = prox_{f_5} = X-B^*u$. It has been shown that both $u^{(t)}$ and $p^{(t)}$ converge linearly and the best convergence rate is attained when ${\mu_t = 2/(\gamma_1+\gamma_2)}$.

In our work, we use the  parallel proximal algorithm (PPXA) proposed by Combettes \emph{et al.} \cite{Combettes_prox} to solve Eq.~(\ref{eqn:convex_opt_chap4}), as this algorithm can be easily implementable on multicore architectures due to its parallel structure. The PPXA  algorithm starts with an initial solution $X^{(0)}$ and computes the $prox_{f_i}, \forall i \in \{1,\ldots, 5\}$ in each iteration and the results are used to update the current solution $X^{(0)}$. The iterative procedure for computing the \emph{prox} of functions $f_i,  \forall i \in \{1,\ldots, 5\}$, and the updating steps are repeated until convergence is reached. The authors have shown that the sequence $(X^{(t)})_{t\geq 1}$ generated by the PPXA algorithm is guaranteed to converge to the solution of problems such as the one given in  Eq.~(\ref{eqn:convex_opt_chap4}). \\

\section{Experimental Results} \label{sec:results}

\subsection{Setup}
   
We study now the performance of our distributed representation scheme for the joint reconstruction of pairs of compressed images. The experiments are carried out in six natural datasets namely, 
\emph{Tsukuba} (views center and right), \emph{Venus} (views 2 and 6) \cite{Scharstein}, \emph{Plastic} (views 1 and 2) \cite{scharstein_plastic}, \emph{Flowergarden} (frames 5 and 6), \emph{Breakdancers} (views 0 and 2) and \emph{Ballet} (views 3 and 4) \cite{MSR_sequence}. The first four datasets have been captured by a camera array where the different viewpoints correspond
to translating the camera along one of the image coordinate axis. In such a scenario, the motion of objects
due to the viewpoint change is restricted to the horizontal direction with no motion along the vertical direction. The depth estimation is thus a one-dimensional search problem and the data cost function given in Eq.~(\ref{eqn:datacost}) is modified accordingly. 

We compress the images independently using an H.264 intra coding scheme; this implementation is carried out using the JM reference software version 18 \cite{JM_software}. The bit rate at the encoder is varied by changing the quantization parameter (QP) in the H.264 coding scheme. In our experiments, we use six different QP parameters, namely $51, 48, 45, 42, 39\; \mbox{and}\; 35$ in order to generate the rate-distortion (RD) plots. Also, we use the same  QP value while encoding the images $\mathcal{I}_1$ and $\mathcal{I}_2$, in order to ensure balanced rate allocation among different cameras. We estimate a depth image from the decoded images $\tilde{I}_1$ and $\tilde{I}_2$ by solving the regularized energy minimization problem of Eq.~(\ref{eqn:energy_chap4}) using $\alpha$-expansion algorithm in Graph Cuts \cite{Graph_cuts}. Unless stated explicitly, we solve the OPT-1 optimization problem with matrix $M$ constructed using Eq.~(\ref{eqn:matrix_M}). The smoothness parameters $(\lambda, \tau)$ of the depth estimation problem of Eq.~(\ref{eqn:depth_est}) and the $(\epsilon_1, \epsilon_2)$ parameters of the OPT-1 joint reconstruction problem are given in Table~\ref{table:parameters} for all the six datasets; these parameters are selected based on trial and error experiments. The solution to the OPT-1 problem is computed by running the PPXA algorithm for 100 iterations. 

We report in this section the performance of the proposed joint reconstruction scheme and highlight the benefit of 
exploiting the inter-view correlation while decoding the images. 
We then study the effect of compression on the quality of the estimated depth images.  
Then, we analyze the importance of the matrix $M$ that enforces correlation consistency only on the corresponding pixels (i.e., the pixels that are not occluded) 
on the quality of the reconstructed images.  Finally, we compare the rate-distortion performance of our scheme w.r.t. state-of-the-art
distributed coding solutions and joint encoding algorithms.

%----- Parameter set for joint representation----------------%
\begin{table}[h!]
\caption{The parameters $(\lambda, \tau)$ in \protect{Eq.~(\ref{eqn:depth_est})} and $(\epsilon_1, \epsilon_2)$ in the OPT-1 problem used in our experiments.} 
\label{table:parameters}
\centering
\begin{tabular}{|c|c|c|c|c|} \hline
{ Dataset} & {$\lambda$}  &{$\tau$}  &{$\epsilon_1$}  & {$\epsilon_2$}  \\ \hline 

Tsukuba      &190  &4  &3  &2\\ \hline

Venus     &220  & 4  &1   &2\\ \hline

Plastic      &120  &4  &1  &2 \\ \hline

Flowergarden   &170  &3  &1  &1.25 \\ \hline

Breakdancer   &300  &160   &2  &1\\ \hline

Ballet   &290  &160   &1 &2.2\\ \hline

\end{tabular}
\end{table}

%--- Venus and Flowergarden datasets: Joint reconstruction performance 
\begin{figure*}[h!]
\centering
$\begin{array}{c@{\hspace{0.1 in}}c} \multicolumn{1}{l}{\mbox{}} &  \multicolumn{1}{l}{\mbox{}} \\
  \epsfxsize=3in \epsffile{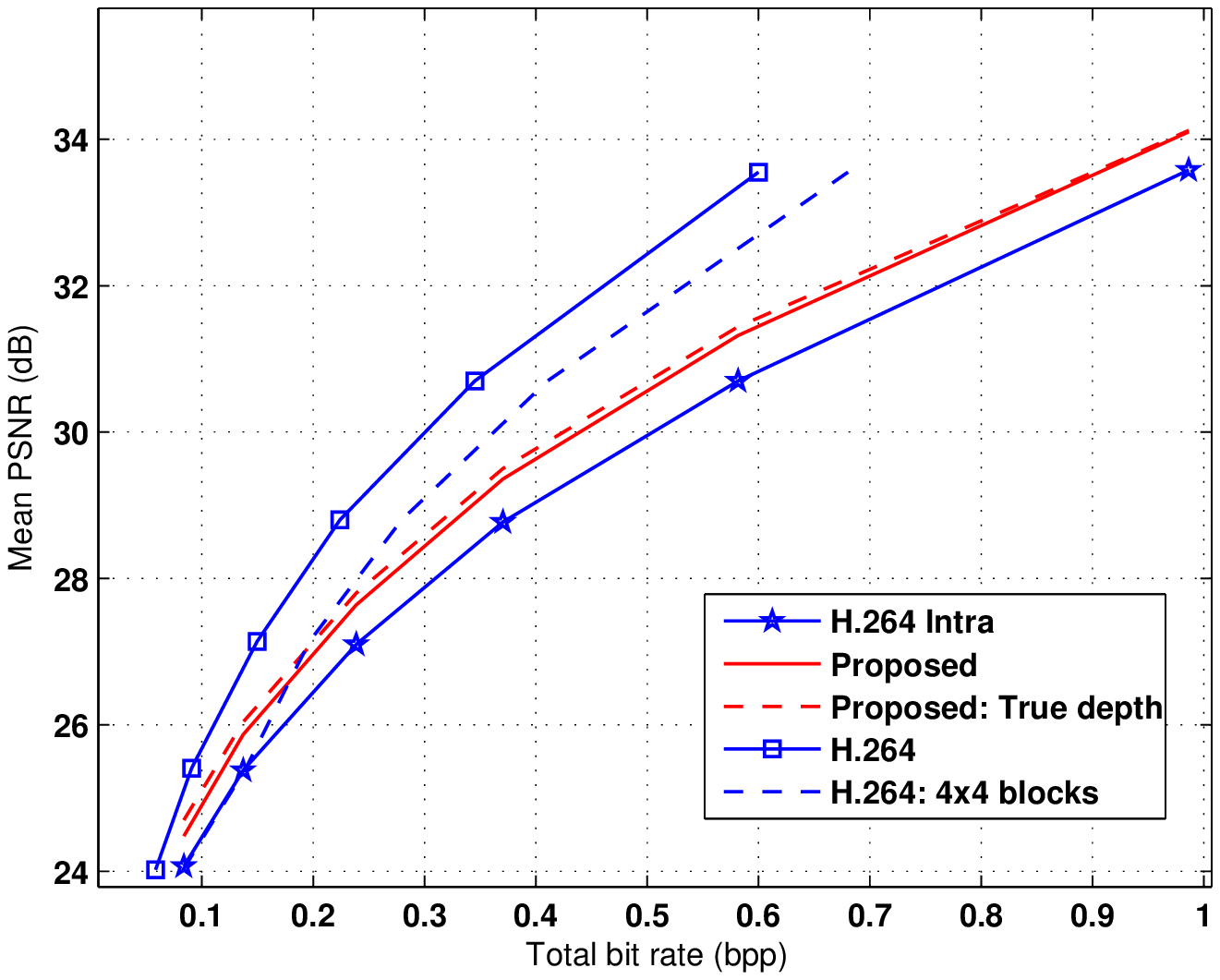} & \epsfxsize=3in \epsffile{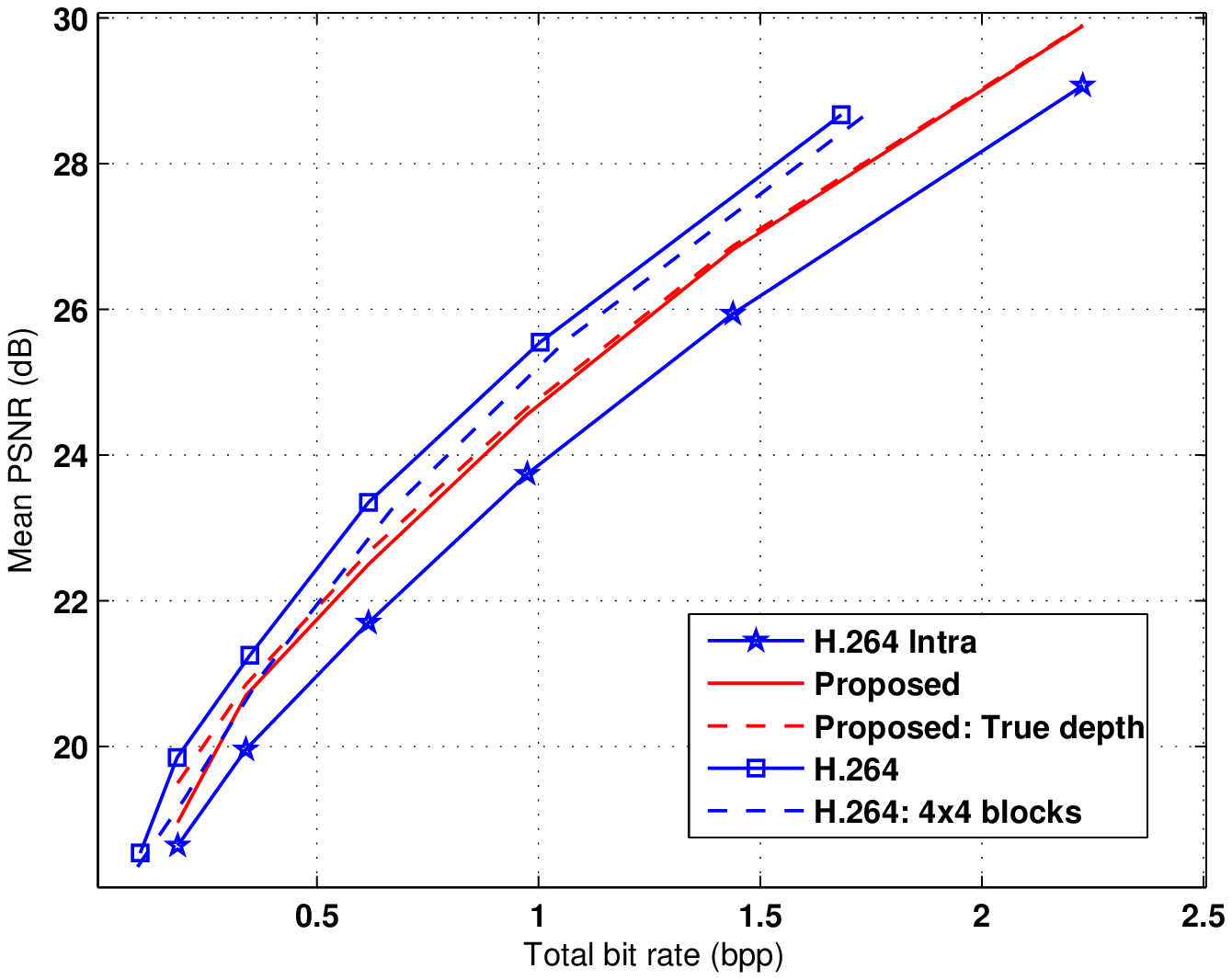} \\
   \mbox{(a)} & \mbox{(b)} \\
  \end{array}$
 \caption{Comparison of the rate-distortion performance between the independent and the joint decoding schemes as well as H.264-based joint encoding schemes: (a) Venus dataset;  and (b) Flowergarden dataset.}
 \label{Fig:jr_ir_rectified}
 \end{figure*}
 
  %--- Breakdancers data set: Joint reconstruction performance 
\begin{figure}[h!]
\centering
  \epsfxsize=3in \epsffile{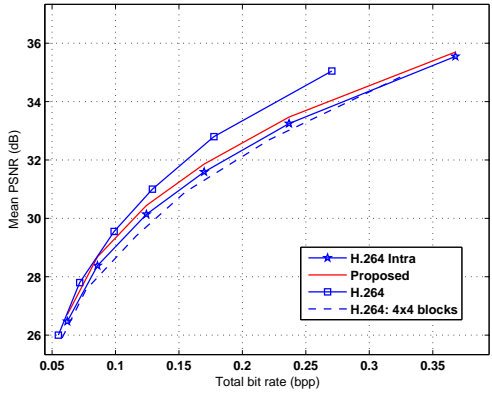}
 \caption{Comparison of the rate-distortion performances between the independent and the joint decoding schemes as well as H.264-based joint encoding schemes for the Breakdancers dataset.}
  \label{Fig:jr_ir_breakdancers}
 \end{figure}

\subsection{Performance Analysis}
We first compare our joint reconstruction results with respect to a scheme where the images are reconstructed independently. Fig.~\ref{Fig:jr_ir_rectified}(a), Fig.~\ref{Fig:jr_ir_rectified}(b) and Fig.~\ref{Fig:jr_ir_breakdancers} compare the overall quality of the decoded images between the independent (denoted as \emph{H.264 Intra}) and the joint decoding solutions (denoted as \emph{Proposed}), respectively for the Venus, Flowergarden and Breakdancers datasets. The x-axis represent the total number of bits spent on encoding the images and the y-axis represent the mean PSNR value of the reconstructed images $\hat{I}_1$ and $\hat{I}_2$. From the plots, we see that the proposed joint reconstruction scheme performs better than the independent reconstruction scheme by a margin of about $0.7$~dB, $0.95$~dB and $0.3$~dB respectively for the different datasets. This confirms that the proposed joint decoding framework is effective in  exploiting the inter-view correlation while reconstructing the images. Similar experimental results have been observed on other datasets. When compared to the first  two datasets, the gain due to joint reconstruction for the Breakdancers dataset is smaller as confirmed in Fig.~\ref{Fig:jr_ir_breakdancers}. It is well known that this dataset is weakly correlated due to large camera spacing \cite{MSR_sequence}, hence the gain provided by the joint decoding is small. 

We then quantitatively compare the RD performances between the joint and the independent coding schemes using the Bjontegaard metric \cite{Bjontegaard}.  In our experiments, we use the first four points in the RD plot for the computation in order to highlight the benefit in the low bit rate region; this corresponds to the QP values $51, 48, 45 \; \mbox{and} \; 42$. The relative rate savings due to joint reconstruction for all the six datasets is available in the second column of Table~\ref{table:ratesavings}. From the values in Table~\ref{table:ratesavings} we see that the benefit of joint reconstruction depends on the correlation among the images; in general, higher the correlation, the better the performance. For example, we see that the Flowergarden dataset gives $22.8\%$ rate savings on average compared to H.264 intra due to very high correlation. On the other hand, the Breakdancers and Ballet datasets only provide about $5\%$ rate savings due to weak correlation mainly because of large distances between the cameras. Though the gain is small for these datasets, we show later that the performance of our scheme competes with the performance of the joint encoding solutions based on H.264 at low bit rates.

 %----- Rate savings for all the datasets ----------------%
\begin{table}[h!]
\caption{Rate savings with respect to the independent coding schemes based on H.264 intra for the stereo images. The average rate savings ($\%$) is computed using the Bjontegaard metric \protect{\cite{Bjontegaard}} for the QP values $52, 48, 45 \; \mbox{and} \; 42$. } 
\label{table:ratesavings}
\centering
\begin{tabular}{|c|c|c|c|c|} \hline
{ Dataset} & {Proposed }  &{Proposed:}  &{H.264: 4x4}  & {H.264} \\
{}  & & {True depth}  &  &   \\ \hline 

Tsukuba      &14.9  &20.5  &15.8   &42.4 \\ \hline

Venus     &15.9  & 21.3  &12.2   &44.9\\ \hline

Plastic      &10.3  &11  &8.4  &28.7 \\ \hline

Flowergarden   &22.8  &29.3  &29.2 &46.5 \\ \hline

Breakdancer   &5.8  &6.6   &-6.5  &8.9\\ \hline

Ballet   &4.4  &6.1   &-8.9 &2.7\\ \hline

\end{tabular}
\end{table}

%--Venus: Disparity error --% 
  \begin{figure*}[h!]
  \centering
 $\begin{array}{@{\hspace{-0.25 in}} c@{\hspace{-0.25 in}}c @{\hspace{-0.25 in}} c@{\hspace{-0.25 in}} c@{\hspace{-0.25 in}} c} \multicolumn{1}{l}{\mbox{}} &  \multicolumn{1}{l}{\mbox{}} &\multicolumn{1}{l}{\mbox{}} &  \multicolumn{1}{l}{\mbox{}} &\multicolumn{1}{l}{\mbox{}} \\
   \epsfxsize=1.6in \epsffile{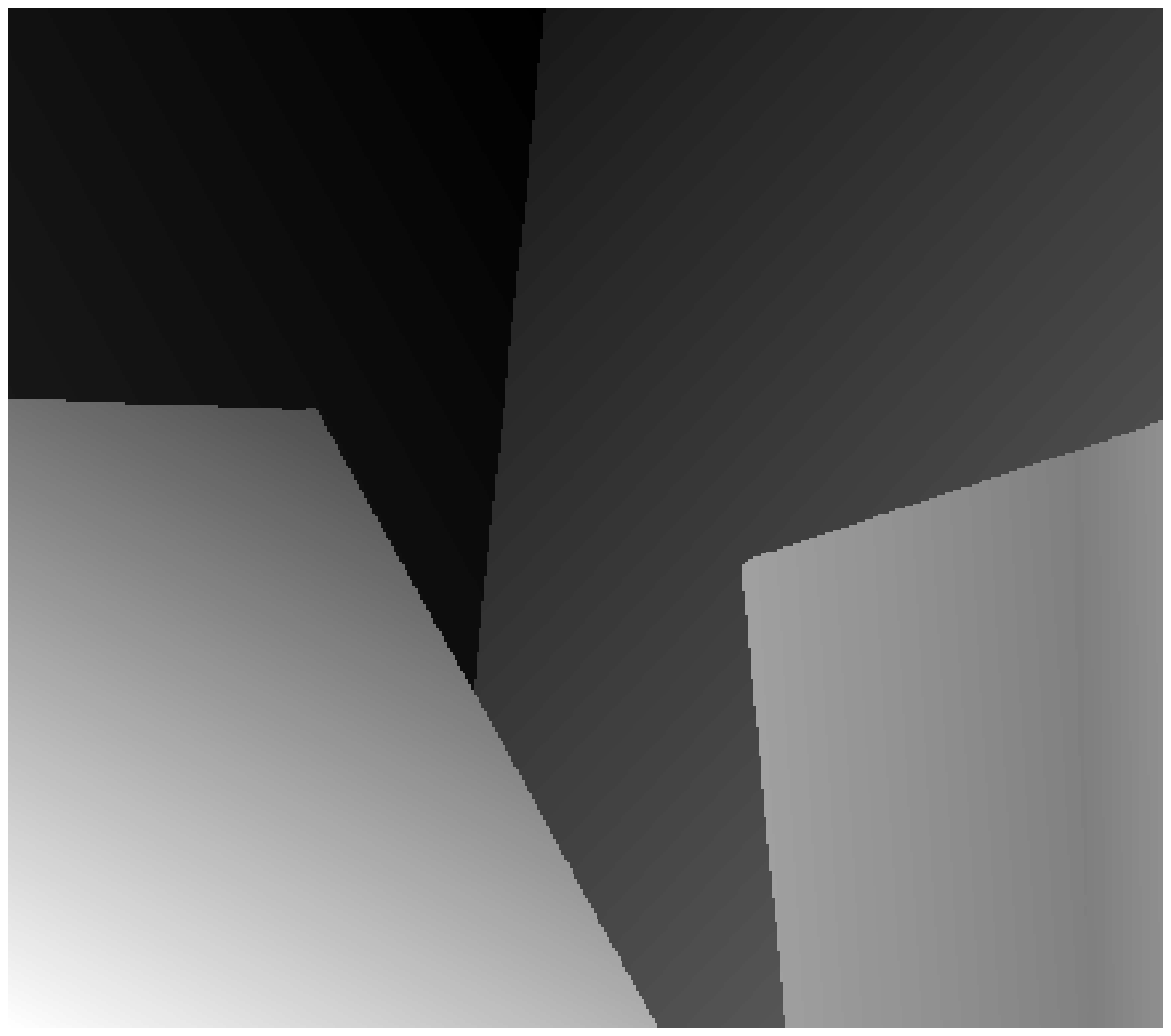} & \epsfxsize=1.6in \epsffile{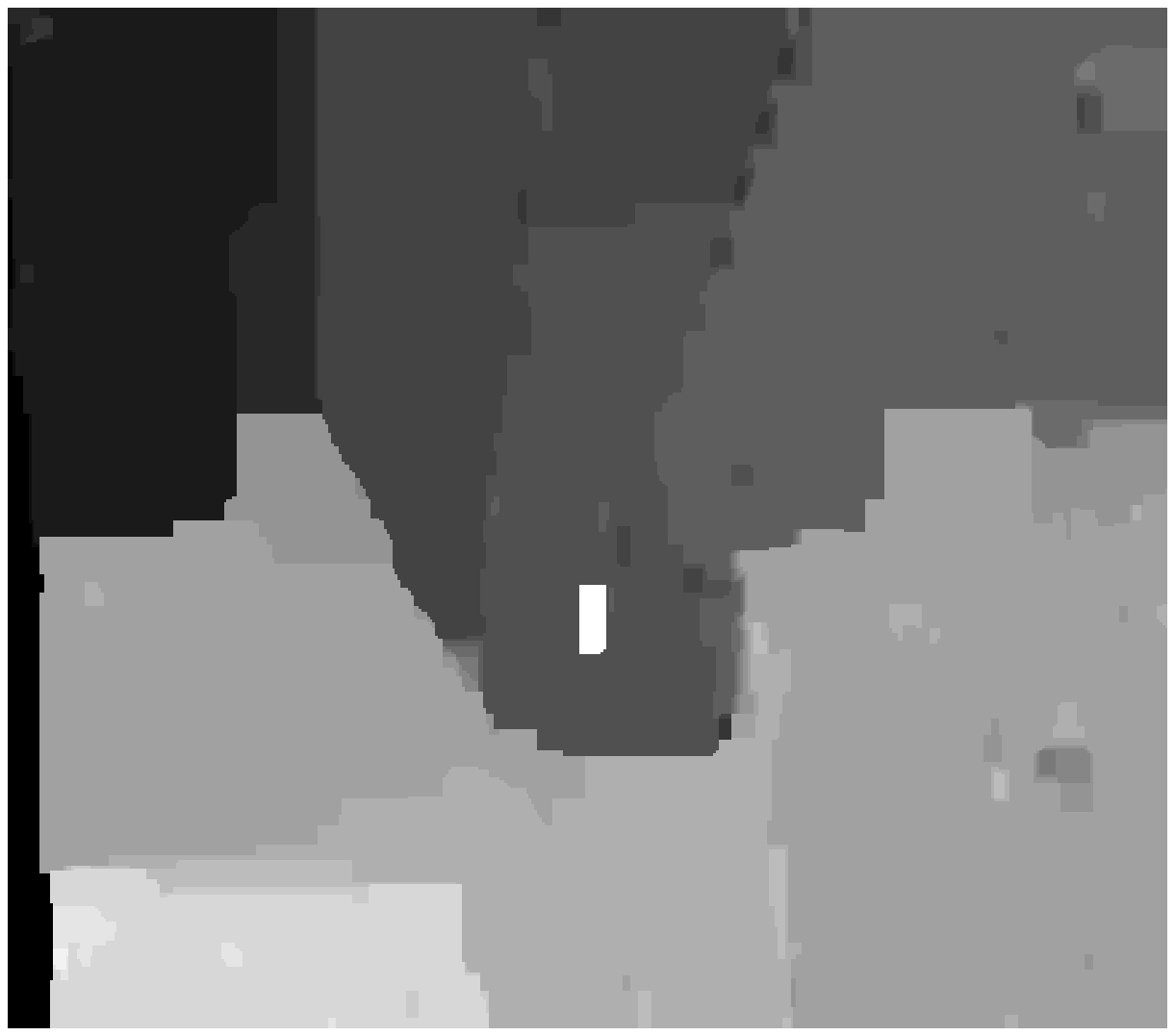} & \epsfxsize=1.6in \epsffile{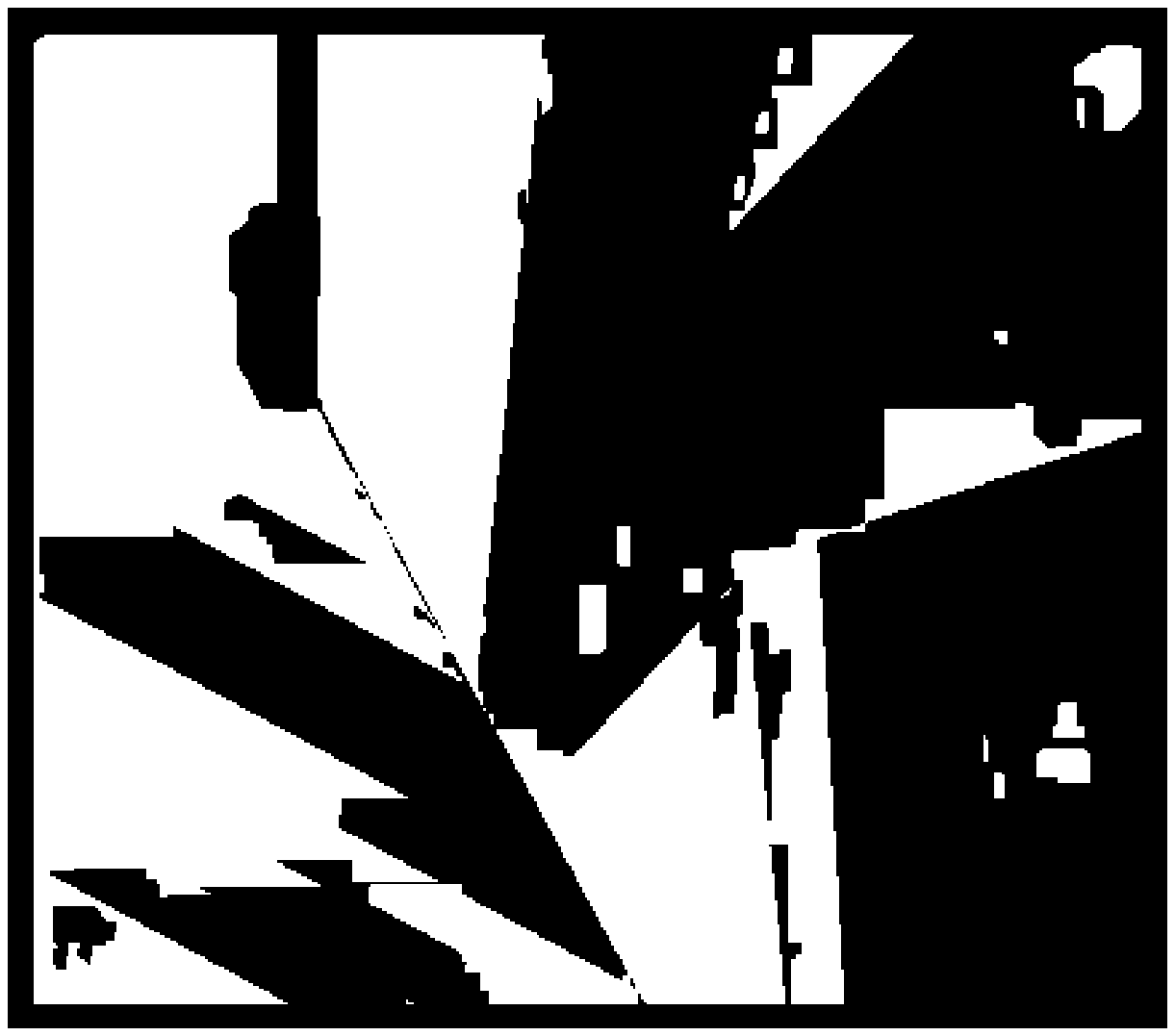} & \epsfxsize=1.6in \epsffile{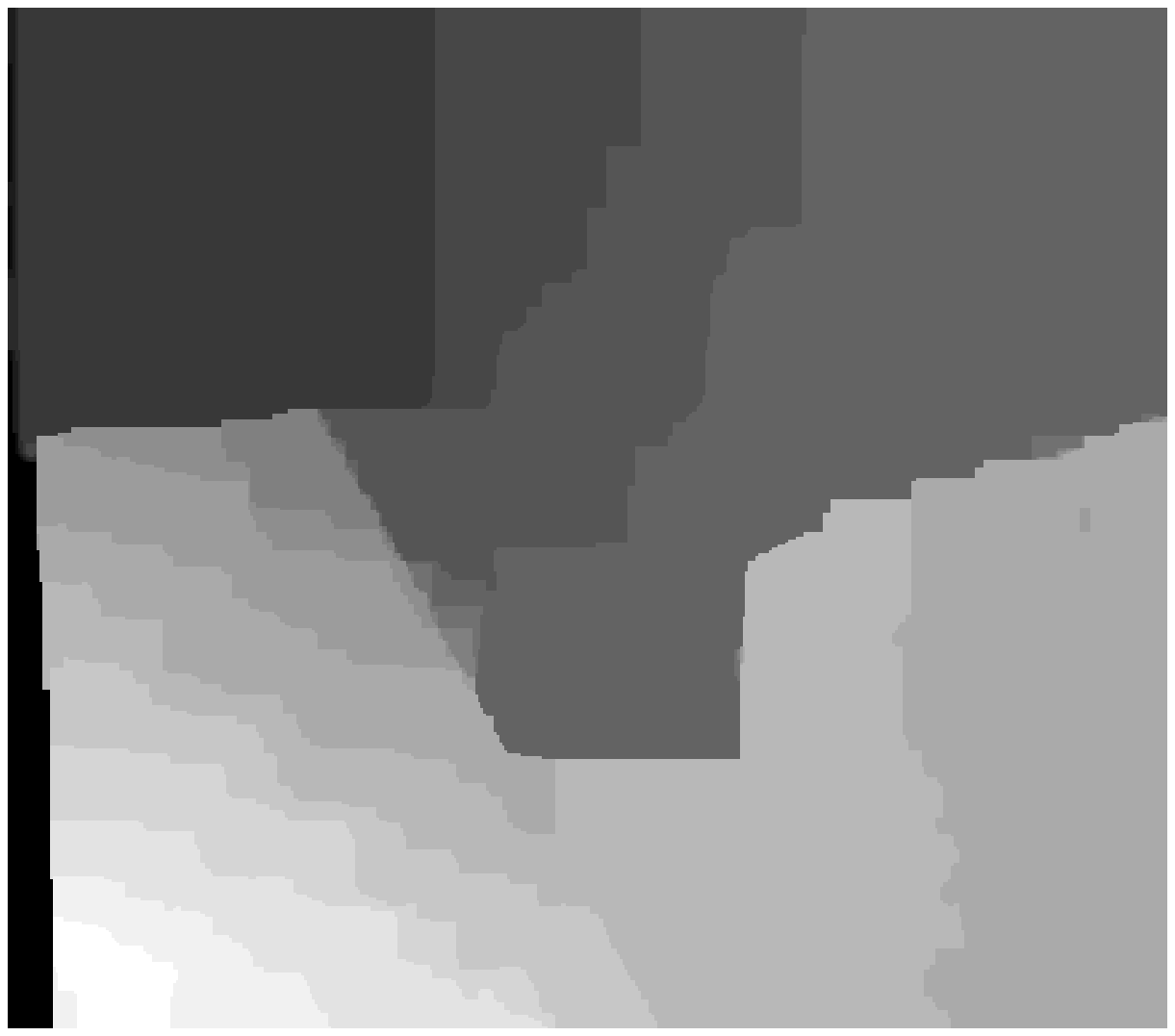} & \epsfxsize=1.6in \epsffile{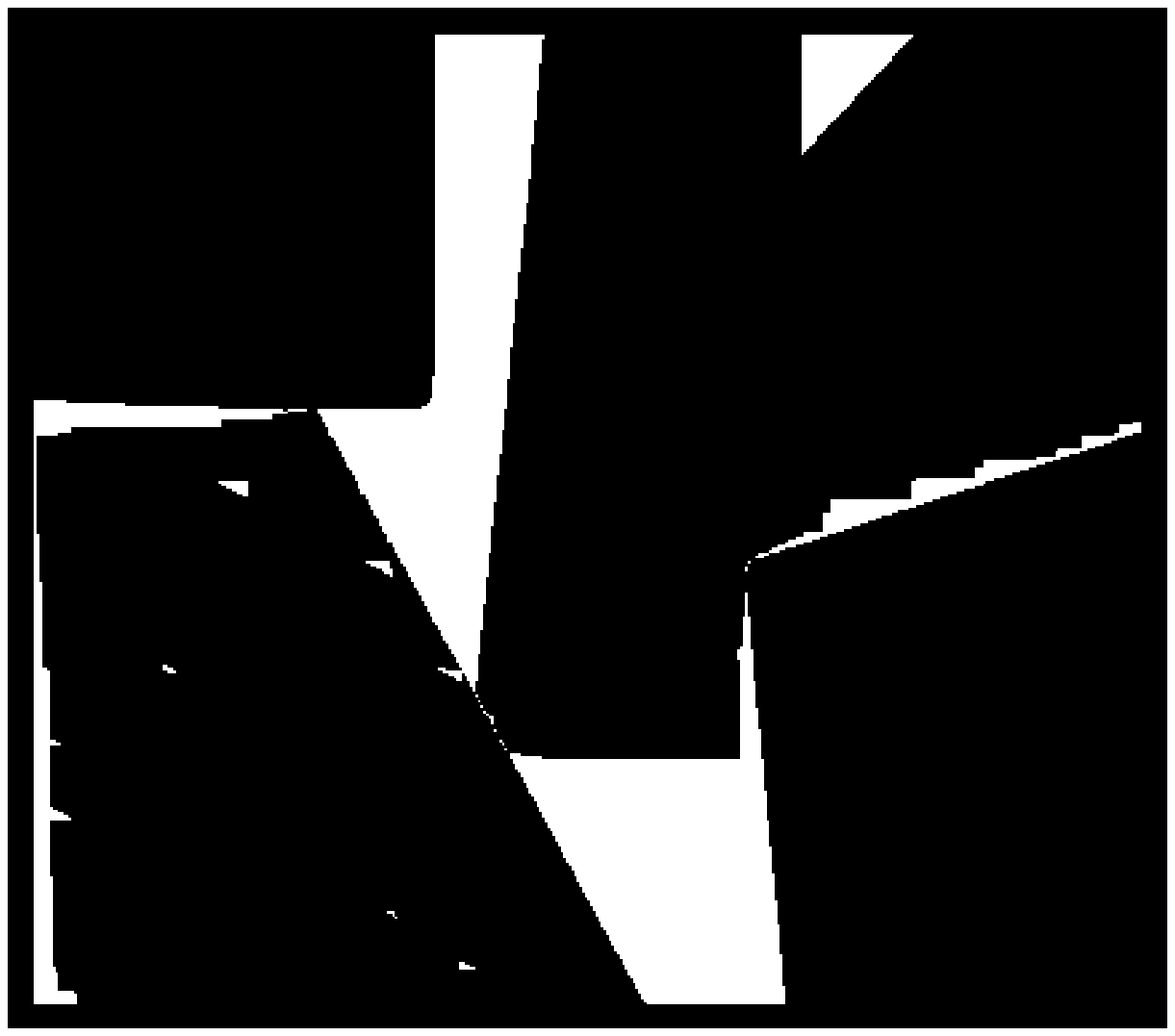} \vspace{-0.1in} \\
   \mbox{(a) $s/D_g$ } & \mbox{(b) $s/D $} & \mbox{(c) $|s/D_g-s/D| >1$} & \mbox{(d) $s/D$} & \mbox{(e) $|s/D_g-s/D| >1$}
  \end{array}$
 \caption{Comparison of the estimated depth image from compressed images with respect to the actual depth information for the Venus dataset. (a) Groundtruth disparity image $s/D_g$; (b) computed disparity image $s/D$ at  QP = 51; (c) disparity error at QP = 51.  The pixels with absolute error greater than one is marked in white. The percentage of white pixels is $43\%$. (d) Computed disparity image $s/D$ at QP = $35$; (e) disparity error at rate at QP = $35$. The percentage of white pixels is $12 \%$.  The parameter $s$ represents the product of the focal length and the baseline distance between the cameras $C_1$ and $C_2$. }
  \label{Fig:venu_deptherror}
    \end{figure*}

   %--flowergarden: Disparity error --% 
  \begin{figure*}[h!]
  \centering
 $\begin{array}{@{\hspace{-0.25 in}} c@{\hspace{-0.25 in}}c @{\hspace{-0.25 in}} c@{\hspace{-0.25 in}} c@{\hspace{-0.25 in}} c} \multicolumn{1}{l}{\mbox{}} &  \multicolumn{1}{l}{\mbox{}} &\multicolumn{1}{l}{\mbox{}} &  \multicolumn{1}{l}{\mbox{}} &\multicolumn{1}{l}{\mbox{}} \\
   \epsfxsize=1.6in \epsffile{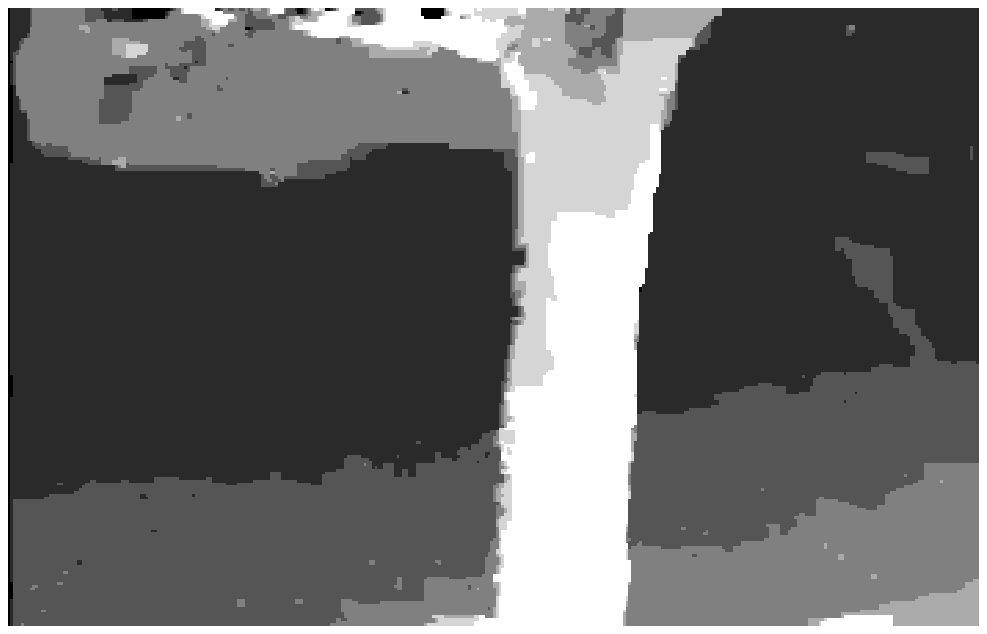} & \epsfxsize=1.6in \epsffile{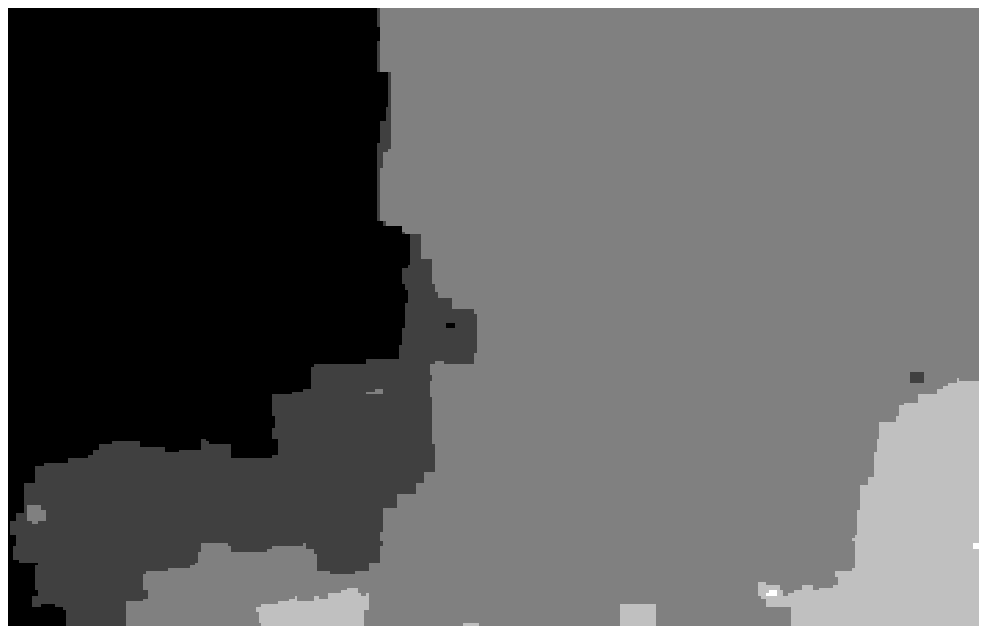} & \epsfxsize=1.6in \epsffile{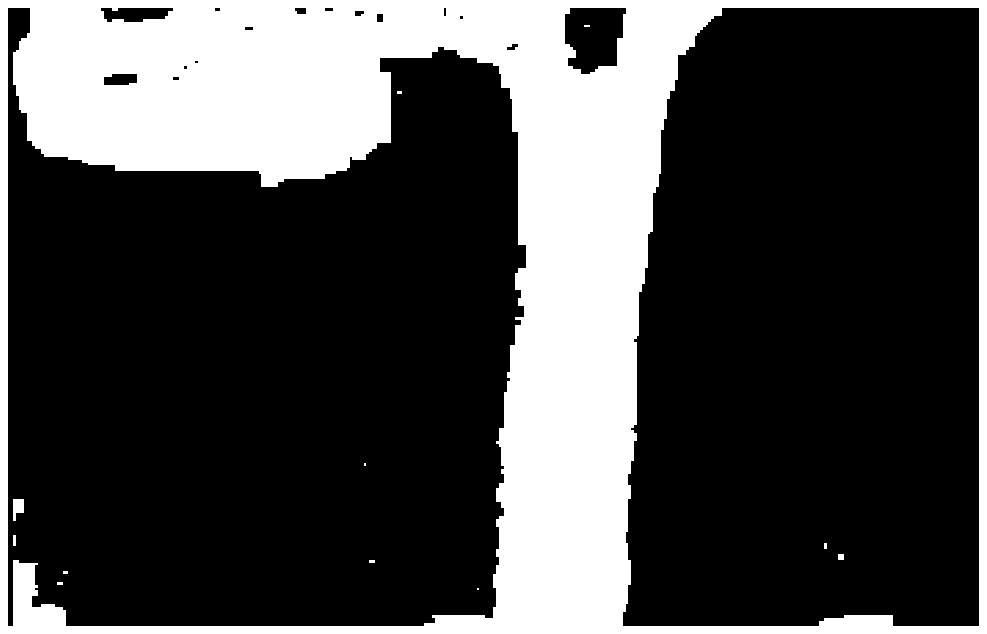} & \epsfxsize=1.6in \epsffile{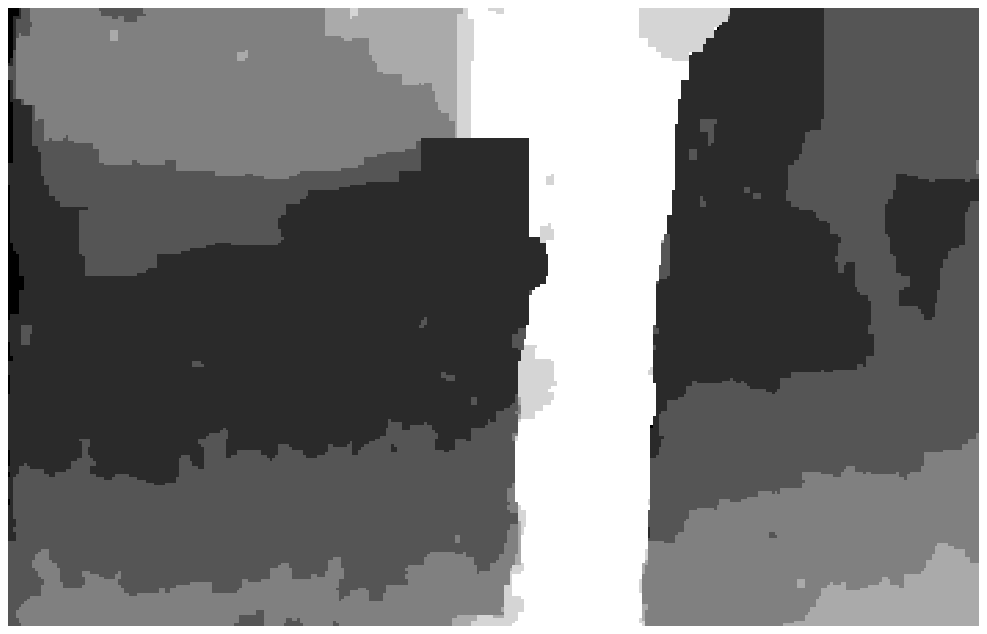} & \epsfxsize=1.6in \epsffile{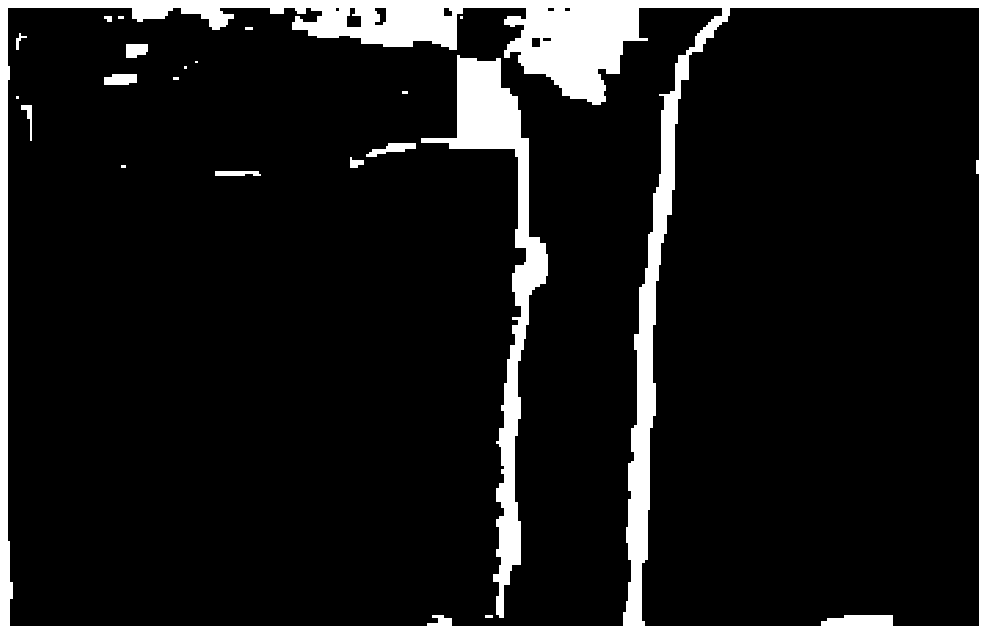} \vspace{-0.1in} \\
   \mbox{(a) $s/D_g$ } & \mbox{(b) $s/D $} & \mbox{(c) $|s/D_g-s/D| >1$} & \mbox{(d) $s/D$} & \mbox{(e) $|s/D_g-s/D| >1$}
     \end{array}$
  \caption{Comparison of the estimated depth image from compressed images with respect to the actual depth information for the Flowergarden dataset. (a) Groundtruth disparity image $s/D_g$; (b) computed disparity image $s/D$ at QP = $51$; (c) disparity error at QP = $51$.  The pixels with absolute error greater than one is marked in white. The percentage of white pixels is $25.3\%$. (d) Computed disparity image $s/D$ at QP = $35$; (e) disparity error at rate at QP = $35$. The percentage of white pixels is $6.6 \%$. The parameter $s$ represents the product of the focal length and the baseline distance between the cameras $C_1$ and $C_2$. }
  \label{Fig:flowergarden_deptherror}
  \end{figure*}

 We then carry out the same experiments in a scenario, where the images are jointly reconstructed using a correlation model that is estimated from the original images. This scheme thus serves as a benchmark for the joint reconstruction, since the correlation is accurately known at the decoder. The corresponding results are denoted as \emph{proposed: True depth} in Fig.~\ref{Fig:jr_ir_rectified}. The corresponding rate savings compared to the independent compression based on H.264 intra is given in the third column of Table~\ref{table:ratesavings}. At low bit rates, in general, we see that our scheme is away from the upper bound performance due to the poor quality of the depth estimation from compressed images.  For example, in Fig.~\ref{Fig:jr_ir_rectified}(b) (for Flowergarden dataset) we see that at bit rate of 0.2 (i.e., QP = $51$), the proposed scheme is away from the upper bound performance by a margin of around $0.5$~dB. As a result, we see in Table~\ref{table:ratesavings} that the rate savings is better, when the actual depth information is used for the joint reconstruction compared to the performance of the scheme where the depth information is estimated from compressed images. We show in Fig.~\ref{Fig:venu_deptherror}(b) and Fig.~\ref{Fig:venu_deptherror}(d) the inverse depth images (i.e., disparity images) estimated from the decoded images $\tilde{I}_1, \tilde{I}_2$ that are encoded with QP = $51$ (resp. total bit rate = 0.08 bpp) and QP = $35$  (resp. total bit rate = 0.98 bpp), respectively for the Venus dataset. Comparing the respective disparity images with respect to the actual disparity information in Fig.~\ref{Fig:venu_deptherror}(a) we observe poor quality disparity results for QP = $51$. Quantitatively, the errors in the disparity images are found to be $43\%$ and $12 \%$, respectively for QP = $51$ and QP = $35$, when it is measured as the percentage of pixels with an absolute error greater than one.  This confirms that the quantization noise in the compressed images are not properly handled while estimating the correlation information.  Similar conclusions can be derived for the Flowergarden dataset from Fig.~\ref{Fig:flowergarden_deptherror}, where, in general, the estimated depth information from highly compressed images is not accurate. Developing robust correlation estimation techniques to alleviate this problem is the target of our future works. We finally see in Fig.~\ref{Fig:jr_ir_rectified} that the reconstruction quality achieved with the correlation estimated from compressed images converges to the upper-bound performance when the rate increases or equivalently, when the quality of decoded images $\tilde{I}_1$ and $\tilde{I}_2$ improves.

 %% importance of operator M 
\begin{figure*}[h!]
\centering
$\begin{array}{@{\hspace{-0 in}}c @{\hspace{0.1 in}}c @{\hspace{0.1 in}}c @{\hspace{0.1 in}}c} \multicolumn{1}{l}{\mbox{}} &  \multicolumn{1}{l}{\mbox{}} &  \multicolumn{1}{l}{\mbox{}}&  \multicolumn{1}{l}{\mbox{}} \\
  \epsfxsize=1.7in \epsffile{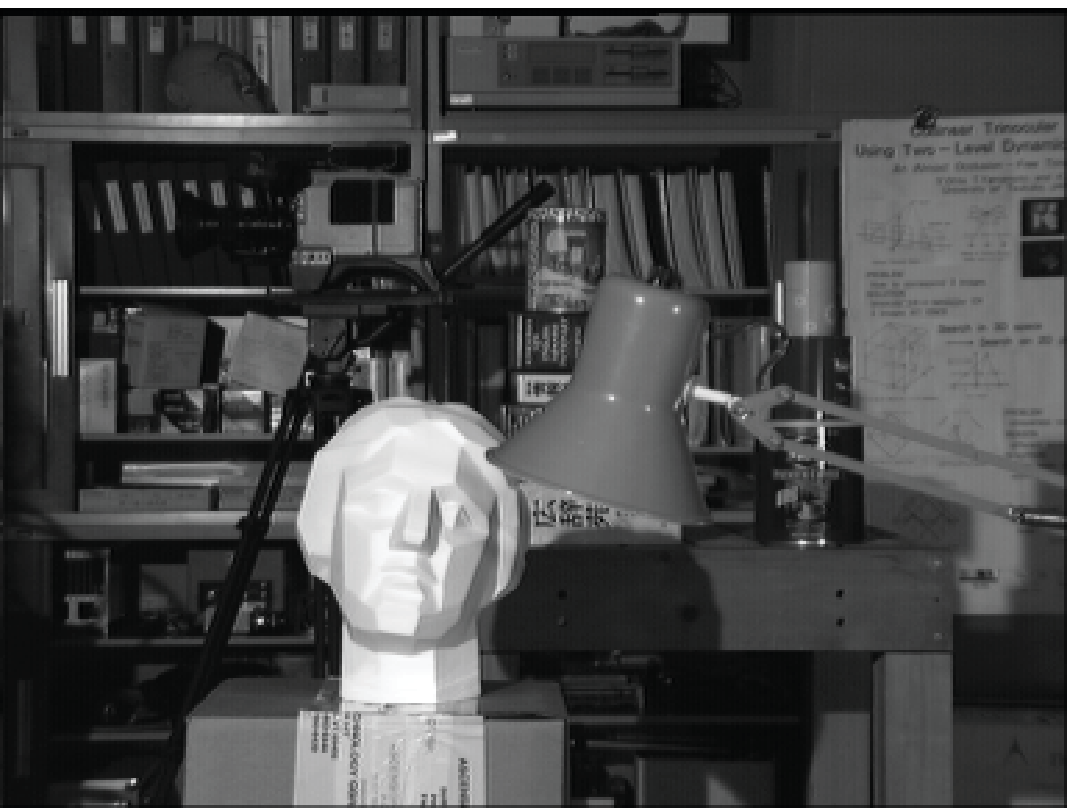} & \epsfxsize=1.7in \epsffile{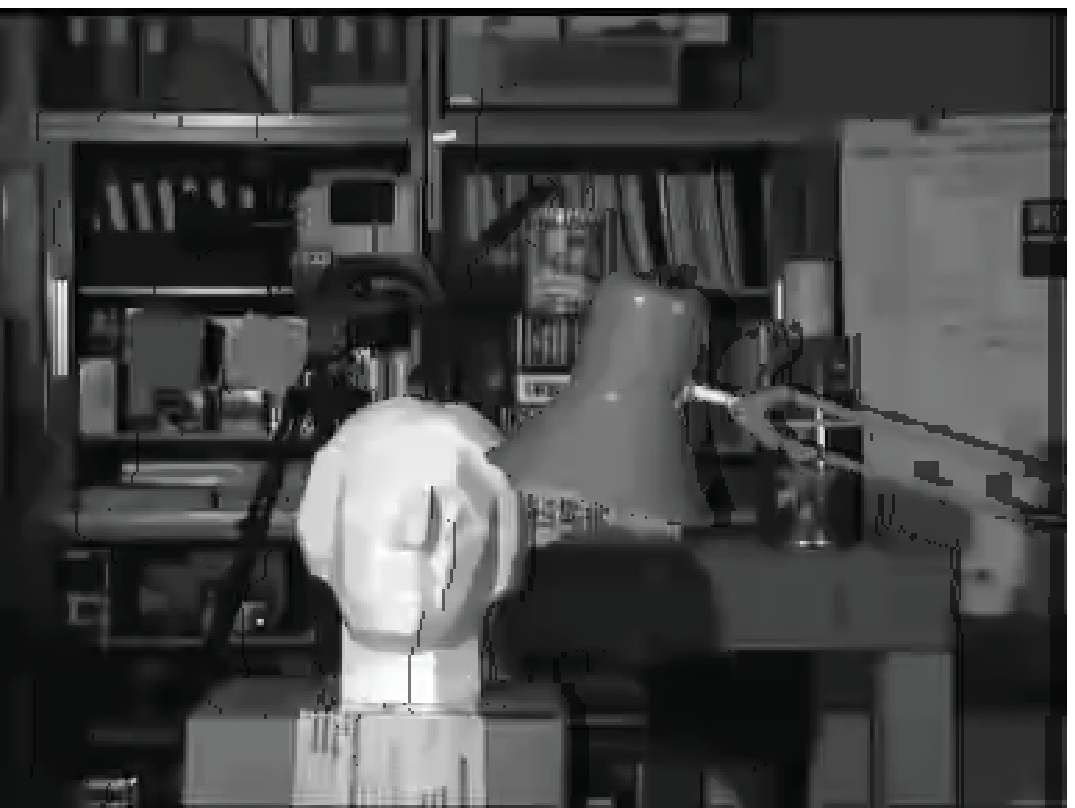} & \epsfxsize=1.7in \epsffile{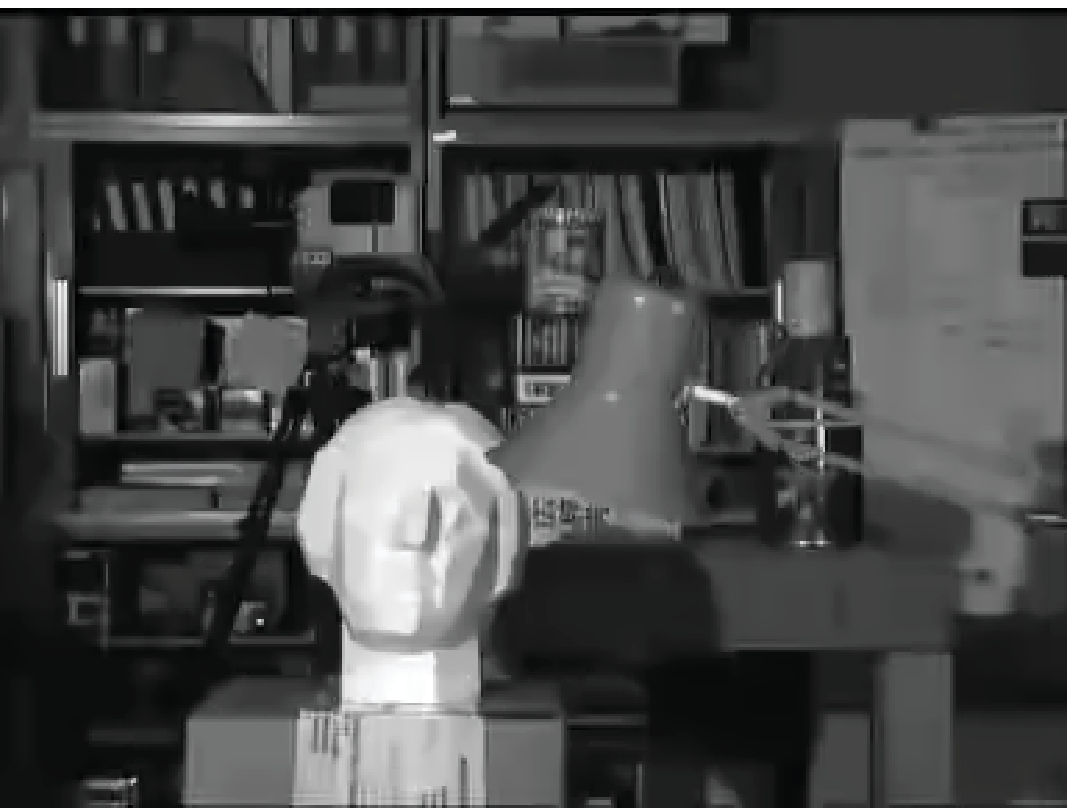}  & \epsfxsize=1.7in \epsffile{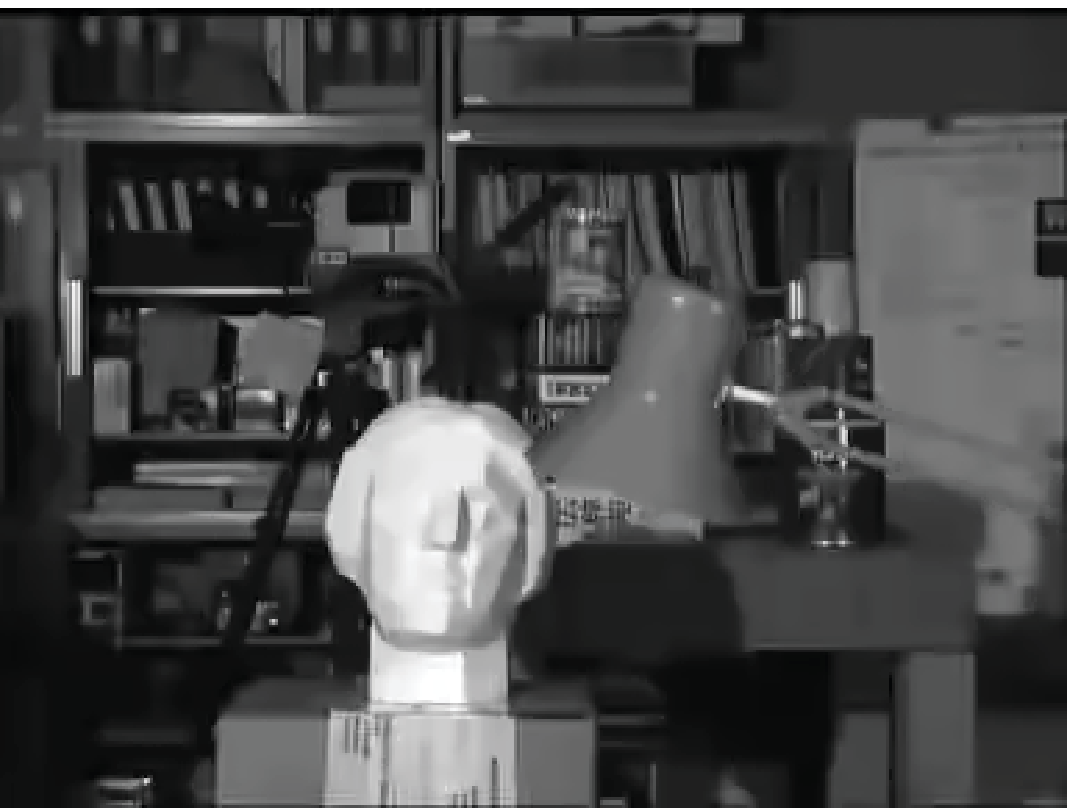}  \\
   \mbox{(a) $\mathcal{I}_2$ }& \mbox{(b) $\hat{I}_2$ } &\mbox{(c) $\hat{I}_2$} & \mbox{(d) $\hat{I}_1$ } 
  \end{array}$
 \caption{Importance of the matrix $M$ in the \ref{eqn:jr_f} optimization problem. (a) Original right image; and (b) reconstructed right image obtained as a solution of the \ref{eqn:jr_f} problem when $M = \mathbbm{1}$. (c) and (d) Reconstructed right and left images, respectively obtained as a solution of the \ref{eqn:jr_f} problem, when the matrix $M$ is constructed based on Eq.~(\ref{eqn:matrix_M}). The PSNR values of the reconstructed images are: (b) $26.84$ dB; (c) $30.01$ dB; and (d) $29.97$ dB. The experiments are carried out in the Tsukuba stereo dataset, where the images are encoded with a QP value of $42$.  }
 \label{Fig:import_M}
\end{figure*}

 %--- Venus and Flowergarden datasets: Joint reconstruction performance w.r.t. davids' scheme
\begin{figure*}[h!]
\centering
$\begin{array}{c@{\hspace{0.1 in}}c} \multicolumn{1}{l}{\mbox{}} &  \multicolumn{1}{l}{\mbox{}} \\
  \epsfxsize=3in \epsffile{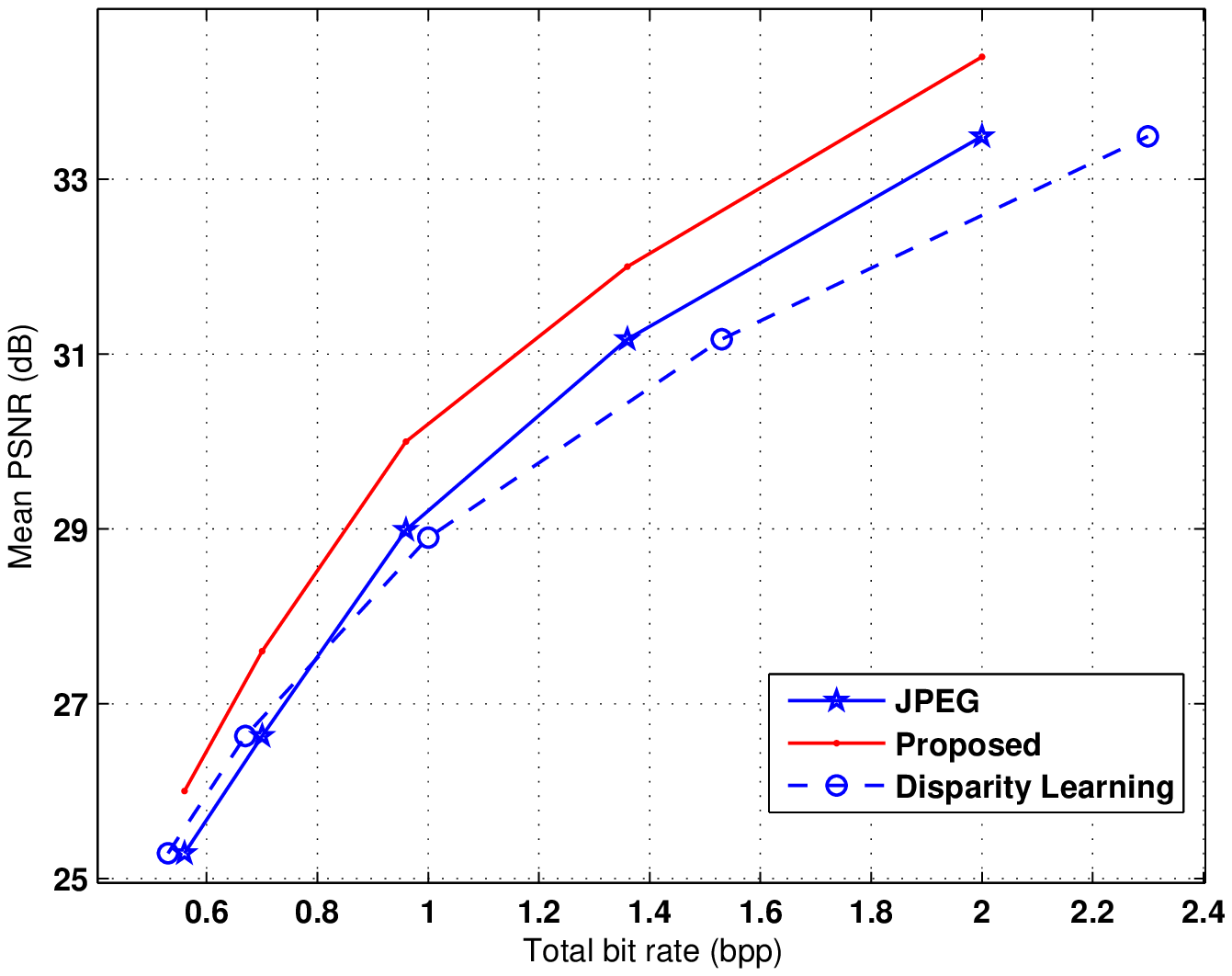} & \epsfxsize=3in \epsffile{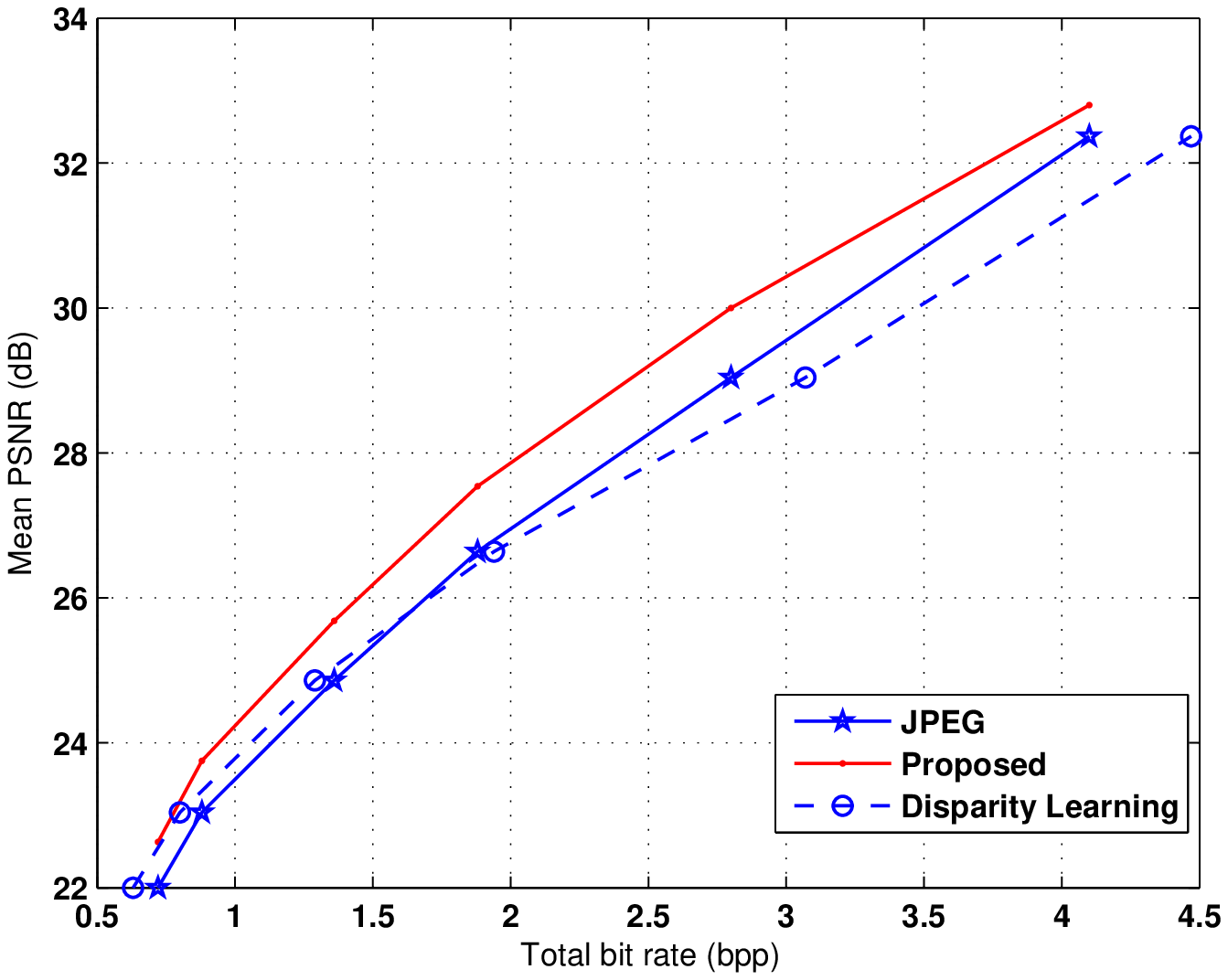} \\
   \mbox{(a)} & \mbox{(b)} \\
     \end{array}$
 \caption{Comparison of the rate-distortion performance between the independent (JPEG) and the joint decoding schemes as well as the DSC scheme based on disparity learning \cite{David}: (a) Venus dataset;  and (b) Flowergarden dataset. In this plot, the independent compression of the images $\mathcal{I}_1$ and $\mathcal{I}_2$ is performed using the JPEG coding scheme.  }
 \label{Fig:jr_jpeg}
 \end{figure*}

We now analyze the importance of the matrix $M$ in the optimization problem \ref{eqn:jr_f} which enables us to measure the correlation consistency objective only to the non-occluded pixels, i.e., the holes in the warped image $A\cdot \mathcal{R}(I_1)$ are ignored while measuring the correlation consistency between the images $A\cdot\mathcal{R}(I_1)$ and $\mathcal{R}(I_2)$. In order to highlight the benefit, we first solve the OPT-1 joint reconstruction  problem by setting $M = \mathbbm{1}$. The corresponding reconstructed right image $\hat{I}_2$ is shown in Fig.~\ref{Fig:import_M}(b). Comparing it with the original right view $\mathcal{I}_2$ in Fig.~\ref{Fig:import_M}(a), we see that the visual artifacts are noticeable in the reconstructed right image $\hat{I}_2$. In particular, we notice strong artifacts along the edges of the \emph{lamp holder} and in the \emph{face} regions; this is mainly due to the improper handling of the occluded pixels. Quantitatively, the PSNR of the reconstructed image $\hat{I}_2$ is $26.84$~dB  (respectively the quality of the reconstructed left view $\hat{I}_1$ is $29.95$~dB). We then solve the OPT-1 optimization problem with a matrix $M$ constructed using Eq.~(\ref{eqn:matrix_M}). The corresponding reconstructed right image $\hat{I}_2$ and left image $\hat{I}_1$ is available in Fig.~\ref{Fig:import_M}(c) and Fig.~\ref{Fig:import_M}(d), respectively. We now do not see any annoying artifacts in the reconstructed image $\hat{I}_2$ due to the effective handling of the occlusions via the matrix $M$. Also, the quality of the reconstructed images becomes quite similar and the respective values for the right and left views are $30.01$~dB and $29.97$~dB.

We then compare the RD performance of our scheme to a distributed coding solution (DSC) based on the LDPC encoding of DCT coefficients, where the disparity field is estimated at the decoder using Expectation  Maximization (EM) algorithms \cite{David}. The resulting RD performance is given in Fig.~\ref{Fig:jr_jpeg}(a) and Fig.~\ref{Fig:jr_jpeg}(b) (denoted as \emph{Disparity learning}) for the Venus and Flowergarden datasets, respectively. In the DSC scheme, the Wyner-Ziv image $\mathcal{I}_2$ is decoded with the JPEG-coded reference image $\mathcal{I}_1$ as the side information.  In order to have a fair comparison between the proposed scheme and this DSC scheme \cite{David}, we carry out our joint reconstruction experiments with the JPEG compressed images.  That is, instead of H.264 intra we now use JPEG for independently compressing the images $\mathcal{I}_1$ and $\mathcal{I}_2$. Then, from the JPEG coded images $\tilde{I}_1$ and $\tilde{I}_2$, we jointly reconstruct a pair of images $\hat{I}_1$ and $\hat{I}_2$ using the methodology described in Section \ref{sec:joint_decoder}. The resulting RD performance of the proposed scheme is available in Fig.~\ref{Fig:jr_jpeg}(a) and Fig.~\ref{Fig:jr_jpeg}(b), respectively for both datasets.  We first notice that the proposed joint reconstruction scheme improves the quality of the compressed images; this is consistent with our earlier observations. We further observe that the disparity learning scheme marginally improves the quality of the compressed images only at low bit rates, however, it fails to perform better than the JPEG coding scheme at high bit rates.  Also, we note that the DSC scheme in \cite{David} requires a feedback channel in order to accurately control the LDPC encoding rate, while our proposed solution does not require any statistical correlation modeling at the encoder nor any feedback channel; this clearly highlights the benefits of the proposed solution. 

For the sake of completeness, we finally compare the performance of our scheme compared to the joint encoding solutions based on H.264. In particular, the joint compression of views is carried out by setting the profile ID = 128; this corresponds to the stereo profile  \cite{JM_software}. In this profile, one of the images (say $\mathcal{I}_1$) is encoded as a I-frame while the remaining view (say $\mathcal{I}_2$) is encoded as a P-frame. We consider two different settings in the H.264 motion estimation, which is performed with a variable and a fixed macroblock size of $4\times4$. The RD performance corresponding to both cases (resp. denoted as \emph{H.264} and \emph{H.264: 4$\times$4 blocks}) is available in Fig.~\ref{Fig:jr_ir_rectified}(a), Fig.~\ref{Fig:jr_ir_rectified}(b) and Fig.~\ref{Fig:jr_ir_breakdancers} for the Venus, Flowergarden and Breakdancers datasets, respectively. Also, we report in the columns 4 and 5 of Table~\ref{table:ratesavings}, the rate savings of the joint encoding scheme compared to the H.264 intra scheme.  First, it is interesting to note that for rectified images (or when the camera motion is horizontal), our scheme competes with the H.264 joint encoding performance when a block size is set to $4\times4$. However, our scheme could not perform as well at high bit rates due to the lack of texture encoding. In other words, our scheme decodes the images by exploiting the geometrical correlation information while the visual information along the texture and edges are not  perfectly captured. However, for the non-rectified images like the Breakdancers dataset (see Fig.~\ref{Fig:jr_ir_breakdancers}), we see that our scheme competes with the joint encoding solutions based on H.264. Similar conclusions can be derived for the Ballet dataset in Table~\ref{table:ratesavings}, where the proposed scheme provides rate savings of 4.4$\%$, while H.264 saves only 2.7$\%$.  This is because, when the images are not rectified, which is the case in the Breakdancers and Ballet datasets, the block-based motion compensation is not an ideal model to capture the inter-view correlation. Also, for the same reason, we see in Fig.~\ref{Fig:jr_ir_breakdancers}  that the H.264 joint encoding with $4\times4$ blocks performs even worse than the H.264 intra coding scheme; this is indicated with a negative sign in Table~\ref{table:ratesavings}.

\section{Joint Reconstruction of multiple images} \label{sec:multiview}

\subsection{Optimization Problem}
So far, we have focused on the distributed representation of pairs of images.  Now, we describe the extension of our framework to datasets with $J$ correlated images $\mathcal{I}_1,\mathcal{I}_2, \ldots, \mathcal{I}_J$ that are captured by the cameras $C_1, C_2, \ldots, C_J$ from different viewpoints. We further assume that the $J$ cameras are calibrated, where we denote the intrinsic camera matrix respectively, for the $J$ cameras as $P_1, P_2, \ldots, P_J$. Also, let $R_1, R_2, \ldots, R_J$ and $T_1, T_2, \ldots, T_J$, respectively represent the rotation and translation of the $J$ cameras with respect to the global coordinate system. Similarly to the stereo setup, the $J$ correlated images $\mathcal{I}_1,\mathcal{I}_2, \ldots, \mathcal{I}_J$ are compressed independently (e.g., H.264 intra or JPEG) with a balanced rate allocation. The compressed visual information is transmitted to the central decoder, where we jointly process all the $J$ compressed views in order to take benefit of the inter-view correlation for improved reconstruction quality. In particular, as carried out in stereo decoding framework, we first estimate a depth image from the $J$ decoded images (resp. $\tilde{I}_1,\tilde{I}_2, \ldots, \tilde{I}_J$) and we use it for joint signal recovery. The $J$ reconstructed images are respectively given as $\hat{I}_1,\hat{I}_2, \ldots, \hat{I}_J$.

We propose to estimate the depth image from the $J$ decoded images in a regularized energy minimization framework as a tradeoff between a data term $\mathcal{E}_{d}$ and a smoothness term $\mathcal{E}_{s}$. The depth image $D$ is estimated by minimizing the energy $\mathcal{E}$ that is represented as 
\begin{equation} \label{eqn:energy_mv}
D = \underset{D_c}{\operatorname{argmin}} \; \mathcal{E}(D_c) = \underset{D_c}{\operatorname{argmin}} \; \{ \mathcal{E}_{d}(D_c) + \lambda \; \mathcal{E}_{s}(D_c)\}.
\end{equation}
where $D_c$ represents the candidate depth images. Note that this formulation is similar to Eq.~(\ref{eqn:energy_chap4}) in the stereo case. 

The data term $\mathcal{E}_{d}(D_c)$ in the multi-view setup should measure the cost of assigning a depth image $D_c$ that is globally consistent with all the compressed images. In the literature, there are plenty of works that address the problem of finding a good multi-view data cost function with global consistency, e.g.,~\cite{mvstereo_overview,multiview_gc, Stretcha}. In this work, for the sake of simplicity, we propose to compute the global photo consistency as the cumulative sum of the data term $E_d(D_c)$ given in Eq.~(\ref{eqn:datacost}). That is, the global photo consistency term is given as 
 \begin{equation} \label{eqn:datacost_mv}
\mathcal{E}_d(D_c) =  \sum_{j=2}^{J}  \sum_{m,n}^{N_1,N_2} \sqnorm{\tilde{I}_j(m,n)- \mathcal{W}_j(\tilde{I}_1(m,n), D_c(m,n))}, 
\end{equation}
where $\mathcal{W}_j$ is the warping function that projects the intensity values in the view $1$ to the view $j$ using the depth information $D_c$. As described previously in Section \ref{sec:depth_est}, this warping is a two step process.  We first project the pixels from view $1$ to the global coordinate system using Eq.~(\ref{eqn:proj_step1}) and then it is projected to the view $j$ using the camera parameters $P_j, R_j$ and $T_j$ (see Eq.~(\ref{eqn:proj_step2})). The objective of the smoothness cost $\mathcal{E}_s$ is to enforce consistency in the depth solution. For a candidate depth image $D_c$, the smoothness energy is computed using Eq.~(\ref{eqn:chap4_energy}). Finally, the minimization problem of Eq.~(\ref{eqn:energy_mv}) can be solved using strong optimization techniques (e.g., Graph Cuts) in order to estimate a depth image $D$ from the decoded images. 
At last, we note that one could estimate a more accurate depth information by considering additional energy terms in the energy model of Eq.~(\ref{eqn:energy_mv}) in order to properly account for the occlusions, global scene visibility, etc. More details are available in the overview paper \cite{mvstereo_overview}. 

Now, we focus on the joint decoding problem, where we are interested in the reconstruction of $J$ correlated images from the compressed information $\tilde{I}_1,\tilde{I}_2, \ldots, \tilde{I}_J$; this is carried out by exploiting the correlation that is given in terms of depth information $D$ or from the operator $A$ derived from the depth $D$ as described in Section \ref{sec:warp_as_linear}. In particular, we can represent the warping operation $\mathcal{W}_j(\tilde{I}_1,D)$ as matrix multiplication of the form $\bar{I}_j = A_j\cdot\mathcal{R}(\tilde{I}_1)$, where $\bar{I}_j$ represents an approximation of the image at viewpoint $j$. We propose to jointly reconstruct the $J$ multi-view images as a solution to the following optimization problem: 
\small
\begin{align} \tag{OPT-2} \label{eqn:jr_mv_n}
(\hat{I}_1, \hat{I}_2, \ldots, \hat{I}_J) =  \underset{I_1, I_2, \ldots, I_J}{\operatorname{argmin}}  \;
& \sum_{j =1}^J \norm{I_j}_{TV}  \\ \nonumber
  \mbox{s.t.} \; 
&  \norm{\mathcal{R}(I_1) -  \mathcal{R}(\tilde{I}_1)}_2 \leq  \delta_1, \\  \nonumber
&  \norm{\mathcal{R}(I_2) - \mathcal{R}(\tilde{I}_2)}_2 \leq  \delta_1, \ldots, \\ \nonumber
&   \norm{\mathcal{R}(I_J) - \mathcal{R}(\tilde{I}_J)}_2 \leq  \delta_1, \\ \nonumber
& \sum_{j=2}^J \sqnorm{M_j(\mathcal{R}(I_j) - A_j\cdot \mathcal{R}(I_1))}  \leq \delta_2,
\end{align}
\normalsize
where $M_j$ (see Eq.~(\ref{eqn:matrix_M})) is a diagonal matrix that is constructed using a similar procedure described in Section~\ref{sec:jr}; this allows to measure the correlation consistency to only to those pixels that are available in all the views. From the above equation, we see that the proposed reconstruction algorithm estimates $J$ TV smooth images that are consistent with both the compressed and the correlation (depth) informations. It is interesting to note that by setting $J=2$ in OPT-2, we get the stereo joint reconstruction problem OPT-1. 

Finally, using the results derived in Prop.~\ref{prop:jr_cvx} it is easy to check that the optimization problem OPT-2 is convex. Therefore, our multi-view joint reconstruction problem OPT-2 can also be solved using proximal splitting methods. We can rewrite the OPT-2 problem as 
\begin{align} \label{eqn:jr_mv_mod}
\underset{X \in \mathbb{R}^{JN}}{\operatorname{argmin}} \; &  \sum_{j =1}^J \norm{\mathcal{R}^{-1}(S_jX)}_{TV} \\  \nonumber
 \mbox{s.t.}  \;
&  {\norm{S_1(Y - X)}_2 \leq  \delta_1},  {\norm{S_2(Y - X)}_2 \leq  \delta_1},  \ldots,  \\   \nonumber
& {\norm{S_J(Y - X)}_2 \leq  \delta_1}, {\norm{HX}_2^2 \leq \delta_2}.  
\end{align}
Here, $X = [ \mathcal{R}({I}_1); \; \mathcal{R}(I_2); \; \cdots \; ; \mathcal{R}(I_J)  ]$, $Y = [  \mathcal{R}(\tilde{I}_1) ;\; \mathcal{R}(\tilde{I}_2) ; \; \cdots \; ;\mathcal{R}(\tilde{I}_J) ]$,  $S_1=[ \mathbbm{1} \; 0  \; \cdots \; 0 ]$, $S_J=[ 0 \: 0  \; \cdots  \;\mathbbm{1}]$,  and the matrix $H$ is given as 
\begin{equation} 
                  H  = {\left[ \begin{array}{ccccc}
               -M_2A_2 & M_2 &0 & \ldots &0  \\
                -M_3A_3 & 0  & M_3 & \ldots &0 \\
                 \vdots & \vdots  & \ddots & \ddots & \vdots  \\
                 -M_JA_J & 0 & 0 & \ldots & M_J \end{array} \right]} .
                 \end{equation}
It can be noted that the above optimization problem is an extension to the one described in Eq.~(\ref{eqn:jr_mod_chap4}), where the TV prior, measurement and correlation consistency objectives are now applied to all the $J$ images. Therefore, the \emph{prox} operators for the objective function and the constraints of Eq.~(\ref{eqn:jr_mv_mod}) can be computed as described in Section \ref{Sec:optmeth}.

\subsection{Performance Evaluation}
We now evaluate the performance of the multi-view joint reconstruction algorithm using five images (center, left, right, bottom and top views) of the
\emph{Tsukuba}  \cite{multiview_gc}, three views (views 0, 1 and 2) of the \emph{Plastic} \cite{scharstein_plastic}, three views (views 0, 2 and 4) of the \emph{Breakdancers} and three views (views 3, 4 and 5) of the \emph{Ballet} \cite{MSR_sequence}.  Similarly to the stereo setup, we independently encode the multi-view images using H.264 intra by varying the QP values. At the joint decoder, we estimate a depth image $D$ from the compressed images by solving Eq.~(\ref{eqn:energy_mv}) with parameters $(\lambda, \tau)$ = $(390,4),(180,4), (330, 180)$ and $(300, 180)$, respectively for the different datasets. Then, using the estimated depth image $D$ we jointly decode the multiple views as a solution to the problem OPT-2 with the matrix $M_j$ constructed using Eq.~(\ref{eqn:matrix_M}). This problem is solved with the parameters $(\delta_1, \delta_2) = (2.5, 7), (1,3), (2.3, 2)$ and $(1.1,4.3)$, respectively for the datasets. Finally, we iterate the PPXA algorithm for 100 times in order to reconstruct the $J$ correlated images.

  %%----- RD plot in Multiview 
\begin{figure*}[h!]
\centering
$\begin{array}{c@{\hspace{0.1 in}}c} \multicolumn{1}{l}{\mbox{}} &  \multicolumn{1}{l}{\mbox{}} \\
  \epsfxsize=3in \epsffile{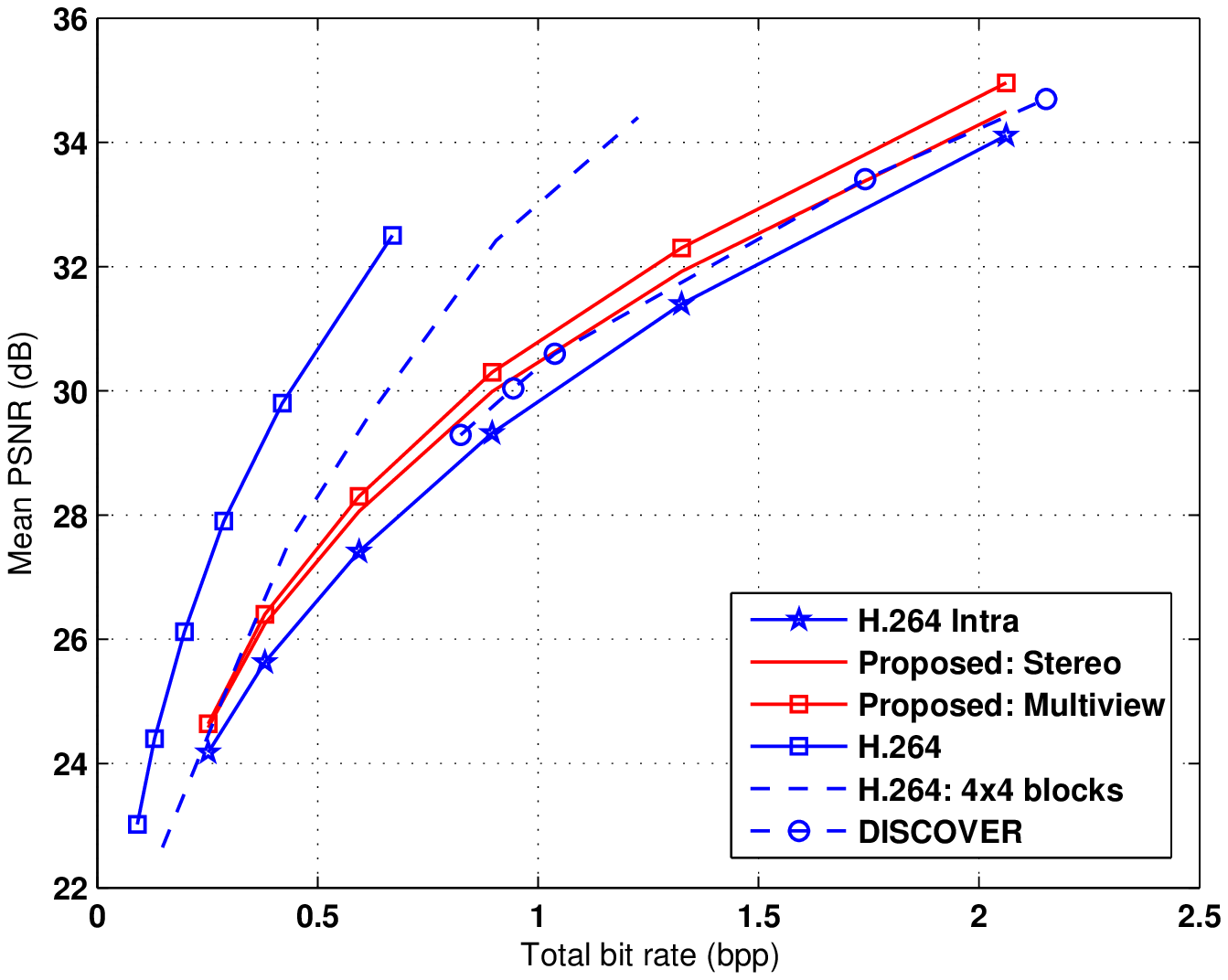} & \epsfxsize=3in \epsffile{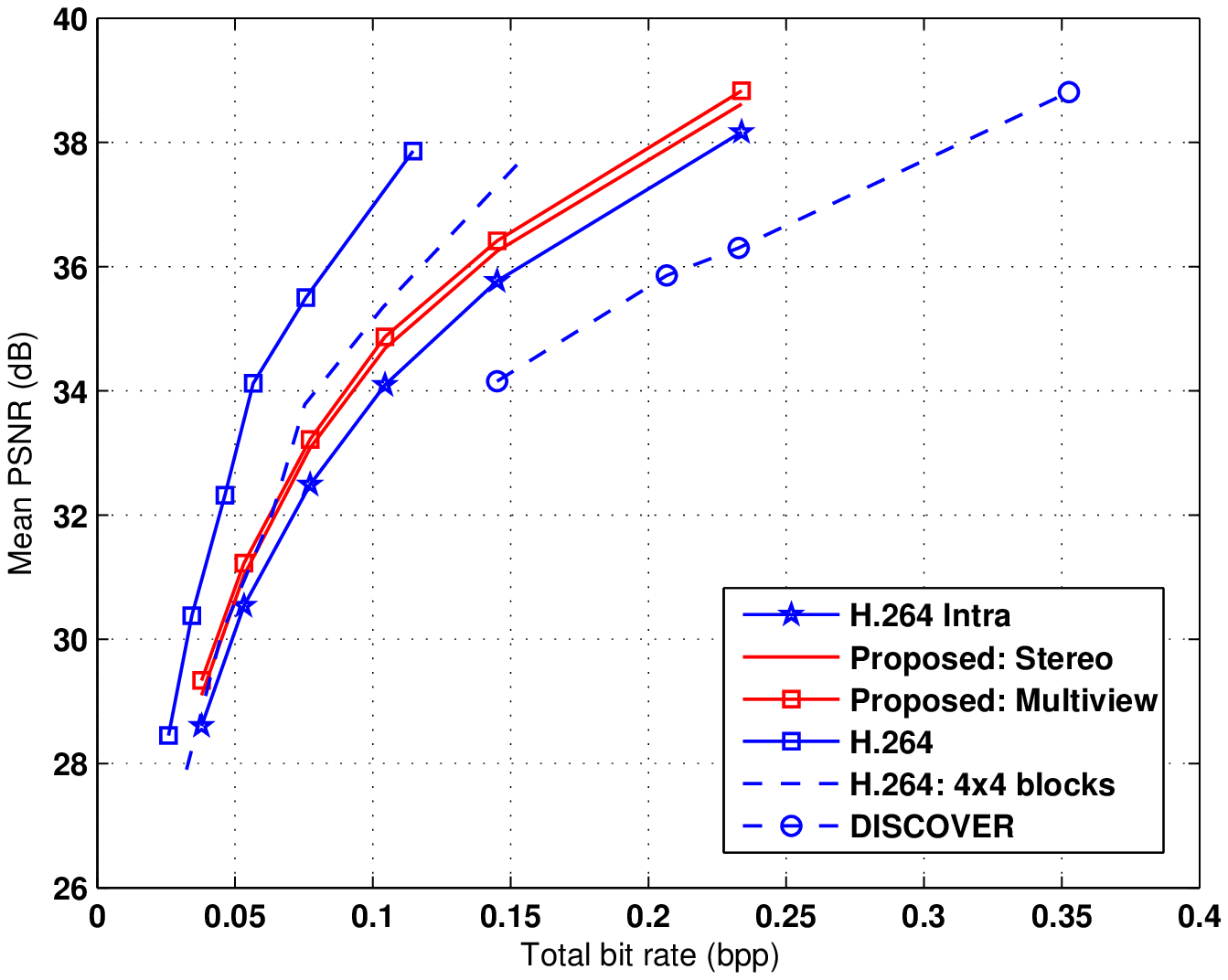} \\
   \mbox{(a)} & \mbox{(b)}
  \end{array}$
 \caption{Comparison of the rate-distortion performance between the independent, proposed, DISCOVER \cite{discover} and H.264-based joint encoding schemes: (a) Tsukuba dataset; and (b) Plastic dataset. The joint reconstruction is performed with $J = 5$ and $J = 3$ views, respectively for the Tsukuba and Plastic datasets. }
 \label{Fig:rd_mv}
\end{figure*}

%%----- RD plot in Multiview 
\begin{figure*}[h!]
\centering
$\begin{array}{c@{\hspace{0.1 in}}c} \multicolumn{1}{l}{\mbox{}} &  \multicolumn{1}{l}{\mbox{}} \\
  \epsfxsize=3in \epsffile{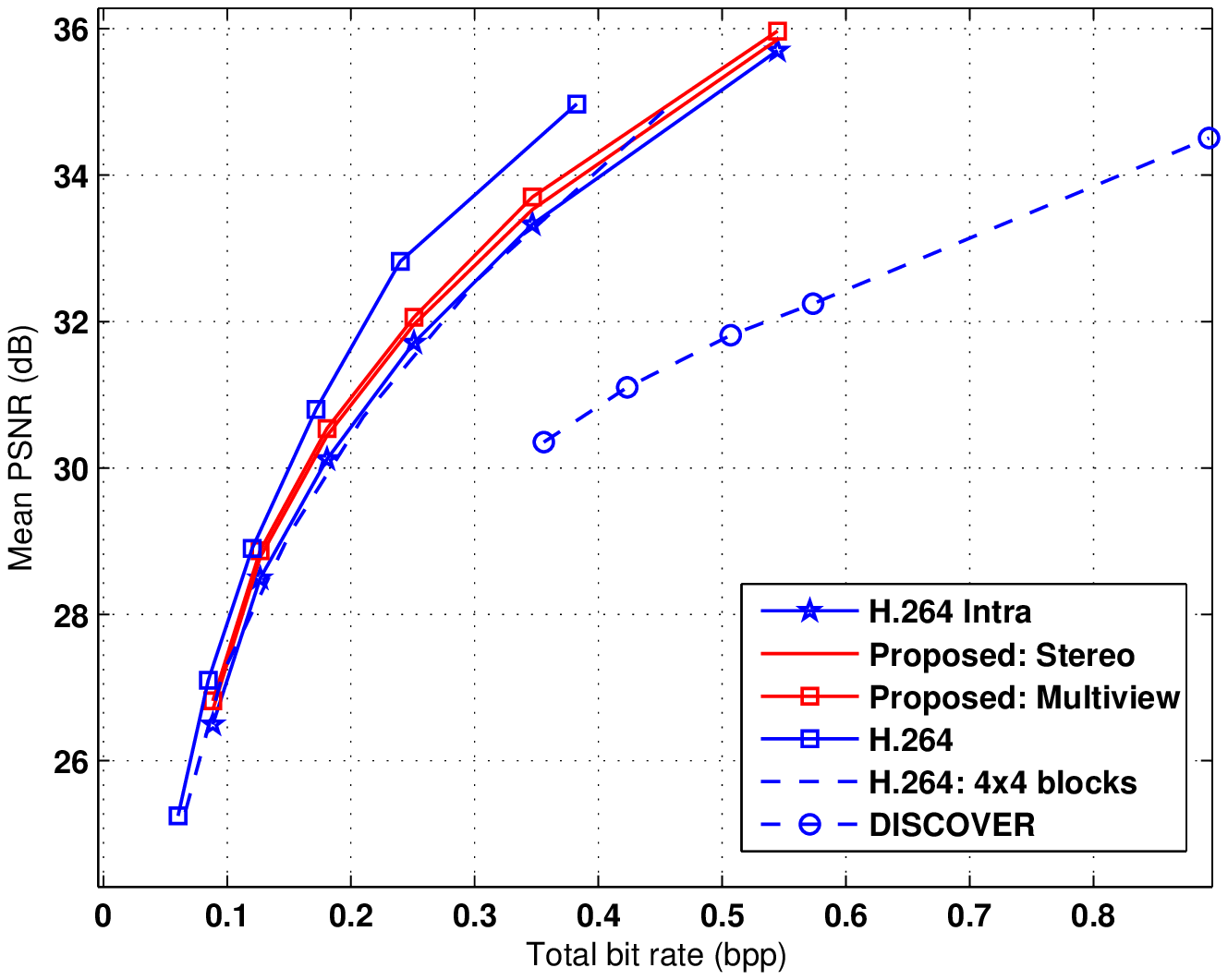} & \epsfxsize=3in \epsffile{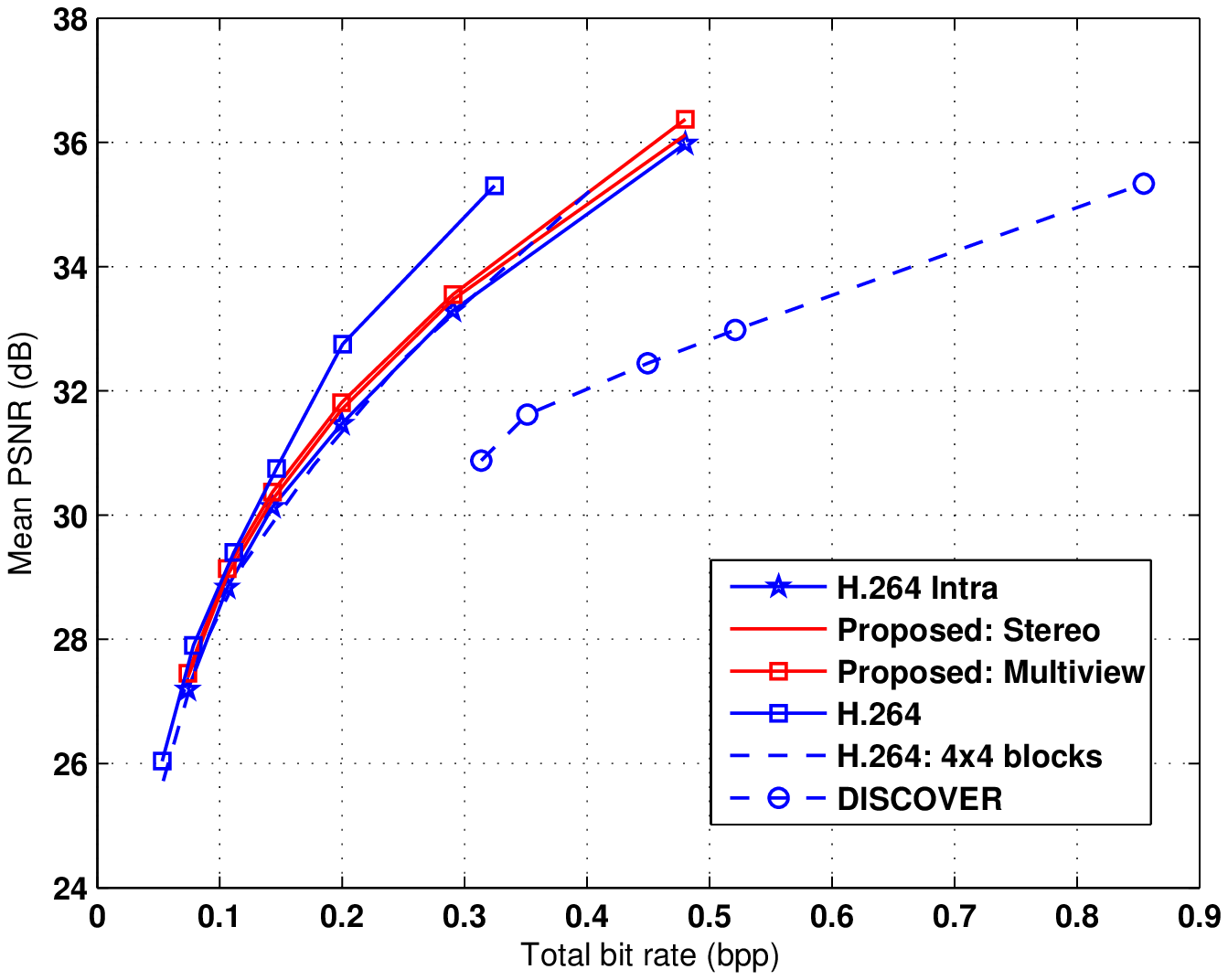} \\
   \mbox{(a)} & \mbox{(b)}
  \end{array}$
 \caption{Comparison of the rate-distortion performance between the independent, proposed, DISCOVER \cite{discover}  and H.264-based joint encoding schemes: (a) Breakdancers dataset; and (b) Ballet dataset. The joint reconstruction is performed with $J = 3$ views. }
 \label{Fig:rd_mv2}
\end{figure*}

We first compare our results with a stereo setup, where the depth estimation and the joint reconstruction steps are carried out with pairs of images. In more details, we take $\mathcal{I}_1$ as being the center image in Tsukuba, the view 1 in Plastic, the view 2 in Breakdancers and the view 4 in Ballet, respectively and we perform joint decoding between the image $\mathcal{I}_1$ and rest of images by selecting different pairs of images independently (all pairs include $\mathcal{I}_1$). For example, for the Tsukuba dataset, we perform the depth estimation and the joint reconstruction steps in the following order: (i) center and right views;  (ii) center and left views; (iii) center and top views; and (iv) center and bottom views. After decoding all the images, we take the mean PSNR of all the reconstructed images. Note that, in this setup the center image is reconstructed four times. For a fair comparison, we keep the reconstructed image $\hat{I}_1$ that gives highest PSNR with respect to $\mathcal{I}_1$. In a similar way, the experiments are carried out for the other datasets, where we perform the joint reconstruction of pairs of images and then compute the average PSNR of the reconstructed images. 
 The resulting RD performance is denoted as \emph{Proposed: Stereo} in Fig.~\ref{Fig:rd_mv}(a), Fig.~\ref{Fig:rd_mv}(b), Fig.~\ref{Fig:rd_mv2}(a) and Fig.~\ref{Fig:rd_mv2}(b) for the different datasets. From Fig.~(\ref{Fig:rd_mv}) and Fig.~(\ref{Fig:rd_mv2})  it is clear that the proposed joint multi-view reconstruction scheme (denoted as \emph{Proposed: Multiview}) performs better than the algorithm where the images are handled in pairs. It clearly highlights the benefits of our proposed solution. We also calculate the rate savings compared to an H.264 intra encoding and the results are tabulated in the second and third columns of Table \ref{table:ratesavings_mv}. It is clear that the rate savings are higher in the multi-view setup than in the stereo setup. Finally, we note that the proposed multi-view joint decoding framework is a simple extension of the stereo image reconstruction algorithm. Still, it permits to show  experimentally that it is beneficial to handle all the multi-view images simultaneously at the decoder rather decoding them by pairs. We strongly believe that the rate-distortion performance in the multi-view problem can be further improved when the depth information is estimated more accurately. For instance, this can be achieved by explicitly considering the visibility and occlusion constraints in the depth estimation framework, e.g., \cite{multiview_gc, Stretcha}. We leave this topic as part of our future work.

 %----- rate savings in multiview ---------------%
\begin{table}[h!]
\caption{Rate savings with respect to the independent coding schemes based on H.264 intra for the multi-view problem. The rate savings $\%$ is computed using the Bjontegaard metric \protect{\cite{Bjontegaard}} for the QP values of $52, 48, 45 \; \mbox{and} \; 42$. }
\label{table:ratesavings_mv}
\centering
\begin{tabular}{|c|c|c|c|c|} \hline
{ Data set} & {Proposed: }  &{Proposed: }  &{H.264: 4x4}  & {H.264}  \\ 
{ } & {Stereo }  &{Multiview}  &{}  & {}  \\\hline 
Tsukuba      &14.7 &19.2  &20.3  &77.8 \\ \hline
Plastic      &10.2  &13.2  &11.5  &45.5 \\ \hline
Breakdancers     &5.7  &7.8  &-1.5  &14.7 \\ \hline
Ballet      &4.2  &6.6  &-2.3  &9.2 \\ \hline
\end{tabular}
\end{table}

We then compare the RD performance of our multi-view joint decoding algorithm to a state-of-the-art distributed coding scheme (DSC) based on the DISCOVER \cite{discover}. The DSC experiments are carried out in the following settings. In the Tsukuba dataset, we consider four views, namely left, right, top and bottom images as the key frames, and the center view is considered as the Wyner-Ziv frame. At the decoder, we generate a side information by fusing two side information images that are generated based on motion compensated interpolation: (i) from the left and right decoded views; and (ii) from the top and bottom decoded views. This fusion step is implemented using the algorithm proposed in \cite{thomas_fusion}. 
For the other datasets, we consider the two extreme views as the key frames and the center view is considered as the Wyner-Ziv frame. In this scenario, a side information image is generated based on motion compensated interpolation from the decoded key frames. The resulting rate-distortion performance is available in Fig.~\ref{Fig:rd_mv} and Fig.~\ref{Fig:rd_mv2} (denoted as \emph{DISCOVER}). Comparing the performance of the proposed scheme (denoted as \emph{Proposed: Multiview}) and the DISCOVER scheme, we show that our scheme outperforms the distributed coding solution. Note that this is the case even in the Tsukuba dataset, where four images are fused together to estimate the best possible side information. Furthermore, we can see that the DSC scheme based on DISCOVER actually performs worse (expect for the Tsukuba dataset) than the H.264 intra scheme where all the images are decoded independently. This is mainly due to the poor quality of the side information image generated based on motion compensated interpolation. In other words, the linear motion assumption is not an ideal model for capturing the correlation between images captured in multi-view camera networks. Finally, it is interesting to note that our joint decoding framework does not require a Slepian-Wolf encoder nor any feedback channel, while the DISCOVER coding scheme requires a feedback channel to ensure successful decoding; this comes at the price of high latency due to multiple requests from the decoder \cite{DVC_overview}. 

For the sake of completeness, we finally compare the performance of our scheme with respect to the joint encoding framework based on H.264 with an $\mbox{IPP}$ coding structure. More precisely, we consider one of the views as the I-frame (this is the views center, 0, 0 and 3 for the different datasets, respectively.), and the remaining views are encoded as P-frames. We perform the joint encoding experiments where the motion compensation is carried out in both variable and fixed block size of $4\times 4$. The resulting rate-distortion performance is available in Fig.~\ref{Fig:rd_mv} and Fig.~\ref{Fig:rd_mv2}. The corresponding rate savings with respect to the H.264 intra are available in columns 4 and 5 of Table \ref{table:ratesavings_mv}. From the plots (see Figs.~\ref{Fig:rd_mv} and ~\ref{Fig:rd_mv2}) and from Table \ref{table:ratesavings_mv}, it is clear that our proposed multi-view reconstruction scheme competes and sometimes beats the performance of H.264 4$\times$4 scheme at low bit rates; this is consistent with the tendencies we have observed in the stereo experiments. However, at high bit rates our scheme performs worse than the H.264 joint coding scheme due to suboptimal representation of high frequency components such as edges and textures. Contrarily to H.264, our scheme is however distributed and this reduces the complexity at the encoders, which is attractive for distributed processing applications.

\section{Conclusions} \label{sec:conc}
In this paper, we have proposed a novel rate balanced distributed representation scheme for compressing the correlated multi-view images captured in camera networks. In contrary to the classical DSC schemes, our scheme compresses the images independently without knowing the inter-view statistical relationship between the images at the encoder. We have
proposed a novel joint decoding algorithm based on a constrained optimization problem that permits to improve the reconstruction quality by exploiting the correlation between images. We have shown that our joint reconstruction problem is convex, so that it can be efficiently solved using proximal methods. Simulation results confirm that the proposed joint representation algorithm is successful in improving the reconstruction quality of the compressed images with a balanced quality between the images. Furthermore, we have shown by experiments that the proposed coding scheme outperforms state-of-the-art distributed coding solutions based on disparity learning and on the DISCOVER. Therefore, our scheme certainly provides an effective solution for distributed image processing with low encoding complexity, since it does not require a Slepian-Wolf encoder nor a feedback channel. Our future work focuses on developing robust techniques to estimate more accurate correlation information from highly compressed images. 

\section{Acknowledgments}
The authors would like to thank Dr. Thomas Maugey for many insightful discussions and for his help in the experimental comparisons  with the DISCOVER distributed coding scheme. 

\bibliographystyle{IEEEtran}
\bibliography{thesis_bibN}

\end{document}